\documentclass{amsart}
\usepackage{amssymb}
\usepackage{graphicx}
\usepackage{multicol}
\usepackage{longtable}
\newtheorem{theorem}{Theorem}[section]

\newtheorem{corollary}[theorem]{Corollary}
\newtheorem{example}[theorem]{Example}
\theoremstyle{definition}

\theoremstyle{remark}

\numberwithin{equation}{section}

\begin{document}
	\title{Construction of Cyclic Codes over a Class of Matrix Rings}
	\author{Soham Ravikant Joshi} 
	\address{Department of Mathematics,	Indian Institute of Technology Patna, Patna-801106}
	\curraddr{}
	\email{E-mail: soham\_2421ma14@iitp.ac.in} 
	\thanks{}
	
	\author{Shikha Patel}
	\address{Department of Mathematics, Indian Institute of Information Technology Bhopal, 462003, India}
	\curraddr{}
	\email{shikha\_1821ma05@iitp.ac.in}
	\thanks{}
	
	\author{Om Prakash$^{\star}$}
	\address{Department of Mathematics, Indian Institute of Technology Patna, Patna-801106}
	\curraddr{}
	\email{om@iitp.ac.in}
	\thanks{* Corresponding author}

	
	\date{}
	
	\dedicatory{}
	
\begin{abstract}
     Let $ \mathbb F_2[u]/ \langle u^k \rangle= \mathbb F_2+u\mathbb F_2+u^2\mathbb F_2+\cdots+u^{k-1}\mathbb F_2 ,$ where $u^k=0$ for a positive integer $k$, and $\mathcal{R}=M_4 (\mathbb F_2( u)/ \langle u^k \rangle)$  be the finite noncommutative non-chain matrix ring of order $4\times4$. This paper presents the construction of cyclic codes over the finite field $\mathbb F_{16}$ via the considered matrix ring $\mathcal{R}$.   In this connection, first, we discuss the structure of the ring $\mathcal{R}$ and show that $\mathcal{R}$ is isomorphic to the ring $( \mathbb F_{16}+ v\mathbb F_{16} + v^2\mathbb F_{16} + v^3\mathbb F_{16}) + u(\mathbb F_{16} + v\mathbb F_{16} + v^2\mathbb F_{16} + v^3\mathbb F_{16}) + u^2(\mathbb F_{16} + v\mathbb F_{16} + v^2\mathbb F_{16}+ v^3\mathbb F_{16}) + \cdots + u^{k-1}(\mathbb F_{16} + v\mathbb F_{16} + v^2\mathbb F_{16} + v^3\mathbb F_{16})$
where  $v^4=0, u^k=0, u^iv^j=v^ju^i$ for $i \in \{1,\dots, k-1\}$ and $j \in \{1, 2, 3\}$. Then, we establish the form of ideals of the ring $\mathcal{R}$ and related cyclic codes over $\mathcal{R}$. Further, we show that these cyclic codes can be written as the direct sums of $\mathcal{R}$-submodules of $\frac{\mathcal{R}[x]}{<x^n-1>}$, and derive the formula for the cardinality of cyclic codes over $\mathcal{R}$. Then, we consider the Euclidean and Hermitian duals of the derived cyclic codes over $\mathcal{R}$. Under the module isometry for $\mathcal{R}$, we use the Bachoc map and the Gray map, which takes derived cyclic code over $\mathcal{R}$ to $\mathbb F_{16}$. Finally, we provide some non-trivial examples of linear codes over $\mathbb F_{16}$ with good parameters that support our derived results and compare a few codes with existing codes in the literature.
\end{abstract}
	

    \maketitle
   \section{Introduction}
    
    Prange \cite{Pron} introduced cyclic codes  in $1957$ and has attracted the attention of several researchers since their elaboration. It is an important class of linear codes and plays a vital role in developing error-correcting codes. Researchers have done numerous revolutionary works on cyclic codes over finite fields and finite commutative rings. Furthermore, many well-known codes, such as BCH, Kerdock, Golay, Reed-Muller, Preparata, Justesen, and binary Hamming codes, are cyclic or derived from it. A cyclic code of length $n$ over a finite field $F$ is the ideal of the quotient ring $F[x]/\langle x^n-1 \rangle$. Being rich with algebraic properties and ease of implementation, researchers prefer this class of codes for encoding and decoding over traditional linear codes. In 1994, Hammons et al. \cite{ss} provided a new direction for linear codes over $\mathbb{Z}_4$ and produced several good nonlinear binary codes as Gray images of linear codes. Subsequently, researchers have shown interest in studying linear codes over finite rings.

    Aside from some notable works on finite commutative rings, codes on noncommutative rings have lately drawn the attention of many researchers; we refer \cite{Alah,G,Luo,F,Pal,W}. In fact, in algebraic constructions of modular lattices \cite{C}, the first concrete set of alphabets $A = M_2(\mathbb{F}_2)$ appeared in 1997. Later, Oggier et al. \cite{F} extended the theory for the noncommutative ring $M_2(\mathbb{F}_2)$ in the setting of space-time codes in 2012. Alahmadi et al. \cite{Alah} investigated cyclic codes over the same ring $M_2(\mathbb{F}_2)$ in 2013. Similarly, Falcunit and Sison \cite{Fal} investigated the structure of cyclic codes on $M_2(\mathbb{F}_p)$ in 2014 and acquired isometric pictures of these codes over the ring $\mathbb{F}_{p^2} + u\mathbb{F}_{p^2}$. Furthermore, Luo and Udaya \cite{Luo} described the structure of cyclic codes over the matrix ring $M_2(\mathbb{F}_2 + u\mathbb{F}_2)$ in 2018 and derived several optimal codes with even lengths over $\mathbb{F}_4$, while Bhowmick et al. \cite{Bandi} addressed cyclic codes and their dual over a matrix ring $M_2(\mathbb{Z}_4)$ in 2018. In addition, Pal et al. \cite{Pal} investigated cyclic codes over the matrix ring $M_4(\mathbb{F}_2)$ in 2019. Recently, Islam et al. \cite{Islam21b} in 2022 extended the previous work to the ring $M_2(\mathbb{F}_p + u\mathbb{F}_p)$, where $p \equiv 2$ or $3 \mod 5$ and explored cyclic codes of length $n$. Then Patel et al. \cite{s} extended this work to cyclic codes over $M_4(\mathbb{F}_{2} + u\mathbb{F}_{2})$ in 2022 itself.
    \section{ Motivation and  Background  }

    Motivated by \cite{F}, here we consider an extension of this setting by studying cyclic codes over the matrix ring $\mathcal{R}$.  In the context of slow fading MIMO (Multiple Input Multiple Output) channels, as considered in the \cite{F} on space-time coded modulations, robust code design is essential for reliable communication. Such channels benefit significantly from codes that ensure full diversity and a non-vanishing determinant. Paper demonstrated that when the inner code has a cyclic algebra structure over a number field, as for perfect space-time codes, an outer code can be designed via coset coding, more precisely, by taking the quotient of the algebra by a two-sided ideal, which leads to matrices over finite alphabets for the outer code. Hence, codes over the matrix ring \( M_4(\mathbb{F}_2) \), when properly structured, offer strong candidates for space-time code design. Specifically, they exploited the cyclic algebraic structure of \( M_4(\mathbb{F}_2) \), using the isomorphism with \( \mathbb{F}_{16} + e\mathbb{F}_{16} + e^2\mathbb{F}_{16} + e^3\mathbb{F}_{16} \), where \( e^4 = 1 \), to define an outer code structure. 
   
   On the other hand, our considered ring $\mathcal{R}$ inherits a similar algebraic structure but enriches the space with additional nilpotent elements, offering a larger alphabet and potential for increased flexibility in code design.  Our structure can similarly be viewed as a module over a non-commutative ring with a cyclic base, as used in the case of $M_4(\mathbb{F}_2)$. We propose that analogous determinant conditions can be developed for codes over $M_4 (\mathbb F_2 [u]/ \langle u^k \rangle)$, though this remains an open problem. The study of such codes thus aims not only to generalize the known results over $\mathbb{F}_{16}$, but also to potentially uncover new families of codes suitable for space-time coded modulation in slow fading MIMO channels.
    
    From \cite{F}, if we consider
    $$e=\begin{pmatrix}
    1&0&0&0\\
    0&0&1&0\\
    1&1&0&0\\
    0&0&1&1
    \end{pmatrix} , ~ \omega=\begin{pmatrix}
    0&1&0&0\\
    0&0&1&0\\
    0&0&0&1\\
    1&1&0&0
    \end{pmatrix},$$
    then $M_4( \mathbb F_2) $ is isomorphic to the ring $\mathbb F_{16}+e \mathbb F_{16}+e^2  \mathbb F_{16}+e^3 \mathbb F_{16}$, where  $(1 + e)^ 4 = 0$ and $\mathbb F_ {16}$ is isomorphic to $\mathbb F _2 [\omega]$. Also,
    $\mathbb F_{16}+e \mathbb F_{16}+e^2  \mathbb F_{16}+e^3 \mathbb F_{16}$ can be written as $\mathbb F_{16}+v \mathbb F_{16}+v^2  \mathbb F_{16}+v^3 \mathbb F_{16}$ with $v^ 4 = (1 + e)^ 4 = 0$. Hence,
    \begin{gather} \label{eq1}
    \mathcal{R} = \sum_{i=1}^{k}u^{i-1}( \mathbb F_{16}+ v\mathbb F_{16} + v^2\mathbb F_{16} + v^3\mathbb F_{16}) \nonumber \\   \text{where}~ v^4=0, u^k=0, u^iv^j=v^ju^i~ \text{for}~ i \in \{1,\dots, k-1\} ~\text{and}~ j \in \{1, 2, 3\}.
    \end{gather}
    Our contributions in this paper are as follows: 
    \begin{enumerate}
\item  We first establish the ring structure of $M_4 (\mathbb F_2 [u]/ \langle u^k \rangle)$ and characterize the cyclic codes over it. Then  we propose the structure of Euclidean and Hermitian dual of the derived cyclic codes over $\mathcal{R}$.
\item  We construct the Gray map and the Bachoc map for the correspondence between  $\mathbb F_{16}^{4kn}$ and $M_4 (\mathbb F_2 [u]/ \langle u^k \rangle)$. Also, we provide some non-trivial examples of linear codes over $\mathbb F_{16}$ with good parameters that support our derived results and compare a few codes with existing codes in the literature.
\end{enumerate}
{The remainder of the  paper is organized as follows: Section $3$ establishes the ring structure of $M_4 (\mathbb F_2 [u]/ \langle u^k \rangle)$ and discuss the construction of cyclic codes over it. Section $4$ discusses the Euclidean dual and the Hermitian dual, respectively, of the derived cyclic codes over $\mathcal{R}$. Section $5$ defines the Gray map and the Bachoc map for the correspondence between  $\mathbb F_{p^3}^{4kn}$ and $ M_3 (\mathbb F_p+u\mathbb F_p+u^2\mathbb F_p+\cdots+u^{k-1}\mathbb F_p ).$ Section $6$ includes some examples of non-trivial codes to support our obtained results. Finally, Section $7$ concludes the paper.}

   \section{Cyclic codes over $\mathcal{R}$ }
	  A cyclic code of length $n$ over $\mathcal{R}$ is also an ideal of $\mathcal{R}_n=\mathcal{R}[x]/ \langle x^n -1 \rangle$ where each codeword $(c_0,c_1,\dots,c_{n-1})$ corresponds to a polynomial $c_0 +c_1x+\cdots+c_{n-1}x^{n-1}$. Let $\mathcal{R}[x]$ be the ring of polynomials over $\mathcal{R}$. For any  $a\in \mathcal{R}$,  $\overrightarrow{a}$ be the polynomial reduction of $a$ modulo $u$ and $v$ as defined in Equation (\ref{eq1}). We consider a polynomial reduction mapping $\Delta : \mathcal{R}[x] \rightarrow \mathbb F_{16}[x]$ defined by
    
    $$\sum_{i=0}^{n-1}a_ix^i \mapsto \sum_{i=0}^{n-1}\overrightarrow{a_i}x^i.$$
    
    Then a polynomial $f(x)$ is called basic irreducible over $\mathcal{R}$ if $\Delta(f(x))$  is irreducible over $\mathbb F _{16}$. If $n$ is odd and $x^n - 1 =\prod_{r=1}^{m} f_r$, where $f_ r$'s are irreducible polynomials in indeterminate $x$ over $\mathbb F _{16}$, then the  polynomial $x^n-1$ can be uniquely factorized into pairwise coprime polynomials over $\mathbb F _{16}$.
    
    \begin{theorem}\label{th1}
    
    Let $n$ be an odd positive integer and $x^n - 1 =\prod_{r=1}^{m} f_r$, where $f_ r$'s are basic irreducible polynomials over $\mathcal{R}$. Then
    
    $$\mathcal{R}_n=\mathcal{R}[x]/ \langle x^n -1 \rangle=\displaystyle\bigoplus_{r=1}^{m} \mathcal{R}[x]/ \langle f_r\rangle.$$
    
    \end{theorem}
    \begin{proof}
    
    It follows by applying the Chinese Remainder Theorem on the right $\mathcal{R}$ - module \\ $\mathcal{R}[x]/ \langle x^n -1 \rangle $~\cite{F,R}.
    
     \end{proof}
    
    Since $ \langle f_r\rangle$ is not a two-sided ideal, we consider $\mathcal{R}[x]/ \langle f_r\rangle$ as a right $\mathcal{R}$-module. Now, we find all the possible forms of the submodules of $\mathcal{R}[x]/ \langle f_r\rangle$.

   \begin{theorem}\label{th2}
   
        If $f$ is a basic irreducible polynomial over $\mathcal{R}$, then the right $\mathcal{R}$-submodules of $\mathcal{R}_f=\mathcal{R}[x]/ \langle f\rangle$ are
        \begin{enumerate}
        
        \item $\{0\},  \langle u^iv^j\rangle$ for $i = \{0,1,2,\dots,k-1\}$, $j = \{0,1,2,3\}$;
        
        \item $ \langle u^i + v\beta_{2i} + v^2\beta_{3i} + v^3\beta_{4i}\rangle$ \text{and} $ \langle u^i , v , v^2 , v^3 \rangle $  \text{for} $i = \{1,2,\dots,k-1\}$;
        
        \item $\langle (u + u^2\alpha_{13} + u^3\alpha_{14} + \cdots + u^{k-1}\alpha_{1k})
        + v( \alpha_{21} + u\alpha_{22} + u^2\alpha_{23} + u^3\alpha_{24} + \cdots + u^{k-1}\alpha_{2k} )\\
        + v^2(  \alpha_{31} + u\alpha_{32} + u^2\alpha_{33}  + \cdots + u^{k-1}\alpha_{3k} )
        + v^3( \alpha_{41} + u\alpha_{42} + u^2\alpha_{43}  + \cdots + u^{k-2}\alpha_{4(k-1)} )\rangle   $;
        
        \item $ \langle v( u + u^2\alpha_{13} + u^3\alpha_{14} + \cdots + u^{k-1}\alpha_{1k} )
        + v^2(  \alpha_{21} + u\alpha_{22} + u^2\alpha_{23} + u^3\alpha_{24} + \cdots + u^{k-1}\alpha_{2k} )\\
        + v^3( \alpha_{31} + u\alpha_{32} + u^2\alpha_{33} + u^3\alpha_{34} + \cdots + u^{k-2}\alpha_{3(k-1)} )\rangle  $;
        
        \item $ \langle v^2(  u + u^2\alpha_{13} + u^3\alpha_{14} + \cdots + u^{k-1}\alpha_{1k} )
        + v^3( \alpha_{21} + u\alpha_{22} + u^2\alpha_{23} + u^3\alpha_{24} + \cdots + u^{k-2}\alpha_{2(k-1)} )\rangle  $;
        
        \item $ \langle  v^3( u + u^2\alpha_{13} + u^3\alpha_{14} + \cdots + u^{k-2}\alpha_{1(k-1)} )\rangle  $,
        
        \end{enumerate}
        where $\alpha_{ij} $, and $\beta_{mn}$ are units in $\mathbb F_{16}[x]$.
	\end{theorem}
    
 	\begin{proof}
    
        Let $I$ be a nonzero right $\mathcal{R}$-submodule of $\mathcal{R}_f$ and $g+\langle f\rangle$ be a nonzero element of $I$. Then $g\notin \langle f\rangle$ and there exist $f_{ij} \in \mathbb F_{16}[x]$ for $i = 1$ to 4 and $j = 1$ to $k$ such that
        \begin{align} \label{eq2}
	\nonumber g = ( f_{11} + uf_{12} + u^2f_{13}  + \cdots + u^{k-1}f_{1k})
        + v( f_{21} + uf_{22} + u^2f_{23}  + \cdots + u^{k-1}f_{2k} )\\ \nonumber
        + v^2(  f_{31} + uf_{32} + u^2f_{33}  + \cdots + u^{k-1}f_{3k} )
        + v^3( f_{41} + uf_{42} + u^2f_{43}  + \cdots + u^{k-1}f_{4k} ).
	\end{align}
    
  	Now, we have two cases.
        \begin{enumerate}
        
        \item [(I)] If $\gcd(f_{11},f)=1$, then there exists $f_{11}^{-1}\in \mathbb F_{16}[x]$ such that $f_{11}f_{11}^{-1} +\langle f\rangle=1+\langle f\rangle$. This implies that $I=\mathcal{R}_f$.
        \item [(II)]  If $\gcd(f_{11},f)=f$, then $g+\langle f\rangle=  (uf_{12} + u^2f_{13} + u^3f_{14} + \cdots + u^{k-1}f_{1k})
        + v( f_{21} + uf_{22} + u^2f_{23} + u^3f_{24} + \cdots + u^{k-1}f_{2k} )
        + v^2(  f_{31} + uf_{32} + u^2f_{33} + u^3f_{34} + \cdots + u^{k-1}f_{3k} )
        + v^3( f_{41} + uf_{42} + u^2f_{43} + u^3f_{44} + \cdots + u^{k-1}f_{4k} ).$ 
		 
	\end{enumerate}
    
        Let $\gcd(f_{1j}, f) = \gcd(f_{2j}, f) = \gcd(f_{3j}, f) = \gcd(f_{4j}, f) = f$ for $j \in \{1,\dots,i\}$,
        and $\gcd(f_{1(i+1)},f)  = 1 $. Then, there exists $f_{1(i+1)}^{-1}$ such that $f_{1(i+1)}f_{1(i+1)}^{-1}$ = 1. Hence, $I = \langle u^i \rangle $ for i = 1 to k-1.
  
        If $\gcd ( f_{1j}, f ) = f $ for $j \in \{1,\dots,k\}$, $\gcd ( f_{2l}, f) = \gcd ( f_{3l}, f) = \gcd ( f_{4l}, f) = f $ for $l \in \{1,\dots, i\}$ and $ \gcd( f_{2(i+1)},f) = 1$, then there exists $ f_{2(i+1)}^{-1}$ such that $ f_{2(i+1)}f_{2(i+1)}^{-1}$ = 1 and $I = \langle vu^i \rangle $. Similarly, we have $\langle v^2u^i \rangle $ and $\langle v^3u^i \rangle $ for $i = 1$ to $k-1$. 
  
        Following the above process, if $\gcd( f_{1j},f ) = f $ for $j \in \{1, \dots, k\}$, and $\gcd (f_{21}, f)$ $ = 1$, then $I = \langle v \rangle $. If $\gcd( f_{1j},f ) = \gcd( f_{2j},f ) = f $ for $j \in \{1, \dots, k\}$, and $\gcd (f_{31}, f) = 1$, then $I = \langle v^2 \rangle $. Also, if $\gcd( f_{1j},f ) = \gcd( f_{2j},f ) = \gcd( f_{3j},f ) = f $ for $j \in \{1, \dots, k\}$, and $\gcd (f_{41}, f) = 1$, then $I = \langle v^3 \rangle $. 
  
        \indent Again, if $\gcd (f_{12},f) = 1$, there exists $f_{12}^{-1}$ such that $f_{12}f_{12}^{-1} =1$. Now, from case (II) we have $u^{k-1}v^3 = gu^{k-2}v^3f_{12}^{-1} \in I$, therefore, $uf_{12}+\cdots+v^3u^{k-2}f_{4(k-1)} = g-v^3u^{k-1}f_{4k}  \in I$. If $\gcd(f_{1j}, f) = 1 $ for $j \in \{3,\dots,k\}$, $\gcd(f_{2j}, f) = \gcd(f_{3j}, f) = 1 $ for $j \in \{1,\dots,k\}$, and $\gcd(f_{4j}, f) = 1 $ for $j \in \{1,\dots,k-1\}$, we have the following cases:
        
        \begin{enumerate}
        \item[(i)] $ \langle (u + u^2\alpha_{13} + u^3\alpha_{14} + \cdots + u^{k-1}\alpha_{1k})
        + v( \alpha_{21} + u\alpha_{22} + u^2\alpha_{23} + u^3\alpha_{24} + \cdots + u^{k-1}\alpha_{2k} )
        + v^2(  \alpha_{31} + u\alpha_{32} + u^2\alpha_{33} + u^3\alpha_{34} + \cdots + u^{k-1}\alpha_{3k} )
        + v^3( \alpha_{41} + u\alpha_{42} + u^2\alpha_{43} + u^3\alpha_{44} + \cdots + u^{k-2}\alpha_{4(k-1)} )\rangle = I  $.
        \item[(ii)] $ \langle (u + u^2\alpha_{13} + u^3\alpha_{14} + \cdots + u^{k-1}\alpha_{1k})
        + v( \alpha_{21} + u\alpha_{22} + u^2\alpha_{23} + u^3\alpha_{24} + \cdots + u^{k-1}\alpha_{2k} )
        + v^2(  \alpha_{31} + u\alpha_{32} + u^2\alpha_{33} + u^3\alpha_{34} + \cdots + u^{k-1}\alpha_{3k} )
        + v^3( \alpha_{41} + u\alpha_{42} + u^2\alpha_{43} + u^3\alpha_{44} + \cdots + u^{k-2}\alpha_{4(k-1)} )\rangle \lneqq I  $. 
        \end{enumerate}
  
        Here, $\alpha_{jl} = f_{jl}f_{12}^{-1}$ are units in $\mathbb{F}_{16}[x]$, where $l \in \{3,\dots,k\}$ for $j = 1$, $l \in \{1,\dots,k\}$ for $j \in \{2,3\}$, and $l \in \{1,\dots,k-1\}$ for $j =4$. Now, from case (II), we have 
        \begin{align*}
        gv &= v(uf_{12} + u^2f_{13} + u^3f_{14} + \cdots + u^{k-1}f_{1k})
        + v^2( f_{21} + uf_{22} + u^2f_{23} \\&
        + u^3f_{24} + \cdots + u^{k-1}f_{2k} ) + v^3(  f_{31} + uf_{32} + u^2f_{33} + u^3f_{34} + \cdots + u^{k-1}f_{3k} ). 
        \end{align*}
        
        Similar to cases (i) - (ii), we have $gv - v^3u^{k-1}f_{3k} =  v(uf_{12} + u^2f_{13} + u^3f_{14} + \cdots + u^{k-1}f_{1k}) + v^2( f_{21} + uf_{22} + u^2f_{23} + u^3f_{24} + \cdots + u^{k-1}f_{2k} ) + v^3(  f_{31} + uf_{32} + u^2f_{33} + u^3f_{34} + \cdots + u^{k-2}f_{3(k-1)} ) ~\in~ I$, and we have the following cases:
        
        \begin{enumerate}
        \item [(iii)]  $  \langle v(u + u^2\alpha_{13} + u^3\alpha_{14} + \cdots + u^{k-1}\alpha_{1k})
        + v^2( \alpha_{21} + u\alpha_{22} + u^2\alpha_{23} + u^3\alpha_{24} + \cdots + u^{k-1}\alpha_{2k} )\\
        + v^3(  \alpha_{31} + u\alpha_{32} + u^2\alpha_{33} + u^3\alpha_{34} + \cdots + u^{k-2}\alpha_{3(k-1)} )\rangle = I$.
        \item [(iv)] $  \langle v(u + u^2\alpha_{13} + u^3\alpha_{14} + \cdots + u^{k-1}\alpha_{1k})
        + v^2( \alpha_{21} + u\alpha_{22} + u^2\alpha_{23} + u^3\alpha_{24} + \cdots + u^{k-1}\alpha_{2k} )\\
        + v^3(  \alpha_{31} + u\alpha_{32} + u^2\alpha_{33} + u^3\alpha_{34} + \cdots + u^{k-2}\alpha_{3(k-1)} )\rangle \lneqq I$.
        \end{enumerate}
      
        Here, $\alpha_{jl} = f_{jl}f_{12}^{-1}$ are units in $\mathbb{F}_{16}[x]$, where $l \in \{3,\dots,k\}$ for $j = 1$, $l \in \{1,\dots,k\}$ for $j = 2$, and $l \in \{1,\dots,k-1\}$ for $j = 3$. Now, from case (II), we have 
        \begin{align*}
        gv^2 &= v^2(uf_{12} + u^2f_{13} + u^3f_{14} + \cdots + u^{k-1}f_{1k})
        + v^3( f_{21} + uf_{22} + u^2f_{23} + u^3f_{24} \nonumber\\&+ \cdots + u^{k-1}f_{2k} ).
        \end{align*} 
        
        Similar to cases (i) - (iv), we have $gv^2 - v^3u^{k-1}f_{2k} =  v^2(uf_{12} + u^2f_{13} + u^3f_{14} + \cdots + u^{k-1}f_{1k}) + v^3( f_{21} + uf_{22} + u^2f_{23} + u^3f_{24} + \cdots + u^{k-2}f_{2(k-1)}) ~\in~I $, and further we have the following cases.
        
        \begin{enumerate}
        \item [(v)] $  \langle v^2(u + u^2\alpha_{13} + u^3\alpha_{14} + \cdots + u^{k-1}\alpha_{1k})
        + v^3( \alpha_{21} + u\alpha_{22} + u^2\alpha_{23}  + \cdots + u^{k-2}\alpha_{2(k-1)}) \rangle = I$.
        \item [(vi)] $  \langle v^2(u + u^2\alpha_{13} + u^3\alpha_{14} + \cdots + u^{k-1}\alpha_{1k})
        + v^3( \alpha_{21} + u\alpha_{22} + u^2\alpha_{23}  + \cdots + u^{k-2}\alpha_{2(k-1)}) \rangle \lneqq I $.
        \end{enumerate}
      
        Here, $\alpha_{jl} = f_{jl}f_{12}^{-1}$ are units in $\mathbb{F}_{16}[x]$, where $l \in \{3,\dots,k\}$ for $j = 1$, and $l \in \{1,\dots,k-1\}$ for $j = 2$. Now, from case (II), we have  
        \begin{align*}
        gv^3 = v^3(uf_{12} + u^2f_{13} + u^3f_{14} + \cdots + u^{k-1}f_{1k}).   
        \end{align*}
        Similarly to cases (i) - (vi), we have $gv^2 - v^3u^{k-1}f_{1k} =  v^3(uf_{12} + u^2f_{13} + u^3f_{14} + \cdots + u^{k-2}f_{1(k-1)}) \in I$, and further we have the following cases.
        
       \begin{enumerate}
        \item [(vii)] $  \langle v^3(u + u^2\alpha_{13} + u^3\alpha_{14} + \cdots + u^{k-2}\alpha_{1(k-1)})\rangle = I $.
        \item [(viii)] $  \langle v^3(u + u^2\alpha_{13} + u^3\alpha_{14} + \cdots + u^{k-2}\alpha_{1(k-1)})\rangle \lneqq I $.
        \end{enumerate}
       
        Here, $\alpha_{jl} = f_{jl}f_{12}^{-1}$ are units in $\mathbb{F}_{16}[x]$, where $l \in \{3,\dots,k-1\}$ for $j = 1$. Next, let $\gcd(f_{1j}, f ) = f$ for $j \in \{2,\dots,i,i+2,\dots,k\}$, and $\gcd(f_{2j},f) = \gcd(f_{3j},f) = \gcd(f_{4j},f) = f $ for $j \in \{2,\dots,k\}$. If we consider $\gcd(f_{l1}, f) = 1 $ for $l \in \{2,3,4\}$ with $ \gcd(f_{1(i+1)},f) = 1$. Then there exist $f_{1(i+1)}^{-1}$ such that $f_{1(i+1)}f_{1(i+1)}^{-1} = 1$. This provides,
    
        \begin{enumerate}
            \item [C1.] $ \langle u^i + v\beta_{2i} + v^2\beta_{3i} + v^3\beta_{4i}\rangle = I$.
            \item [C2.] $ \langle u^i + v\beta_{2i} + v^2\beta_{3i} + v^3\beta_{4i}\rangle \lneqq I $.
        \end{enumerate}
        
        Here, $\beta_{2i} = f_{21}f_{1(i+1)}^{-1}$, $\beta_{3i} = f_{31}f_{1(i+1)}^{-1}$, and $\beta_{4i} = f_{41}f_{1(i+1)}^{-1}$ are units in $\mathbb{F}_{16}[x]$. Now, from case C2., we have the following subcases.
        
        \begin{enumerate}
              \item [(i)] There exists  $u^i + v\beta_{2i}' + v^2\beta_{3i} + v^3\beta_{4i}$ where $\beta_{2i}' \in \mathbb F_{16}[x]$ such that it does not belong to  $ \langle u^i + v\beta_{2i} + v^2\beta_{3i} + v^3\beta_{4i}\rangle.$  This implies that $v(\beta_{2i}-\beta_{2i}')\in I$. But $(\beta_{2i}-\beta_{2i}') \notin I$, therefore, $v \in I$, and $I = \langle u^i, v \rangle $.
              \item [(ii)]  There exists  $u^i + v\beta_{2i} + v^2\beta_{3i}' + v^3\beta_{4i}$ where $\beta_{3i}' \in \mathbb F_{16}[x]$ such that it does not belong to  $ \langle u^i + v\beta_{2i} + v^2\beta_{3i} + v^3\beta_{4i}\rangle.$  This implies that $v(\beta_{3i}-\beta_{3i}')\in I$. But $(\beta_{3i}-\beta_{3i}') \notin I$, therefore, $v^2 \in I$, and $I = \langle u^i, v^2 \rangle $.
              \item [(iii)] There exists $u^i + v\beta_{2i} + v^2\beta_{3i} + v^3\beta_{4i}'$ where $\beta_{4i}' \in \mathbb F_{16}[x]$ such that it does not belong to  $ \langle u^i + v\beta_{2i} + v^2\beta_{3i} + v^3\beta_{4i}\rangle.$  This implies $v(\beta_{4i}-\beta_{4i}')\in I$. But $(\beta_{4i}-\beta_{4i}') \notin I$, therefore, $v^3 \in I$ and $I =   \langle u^i, v^3 \rangle $.
        \end{enumerate}
        
        Finally, from the above subcases (i), (ii) and (iii), we have $ I = \langle u^i , v , v^2 , v^3 \rangle $. 
        \end{proof}

        Let $f_i$ be a factor of $x^n-1$ and consider $\hat{f_i}=\frac{x^n-1}{f_i}$. Then we have the following result for $\mathcal{R}$-modules in $\frac{\mathcal{R}[x]}{\langle x^n-1\rangle}$.
        
        \begin{corollary}\label{cor1}
	Let $x^n-1=f_1f_2\cdots f_m$ where $f_j(1\leq j\leq m)$ are irreducible pairwise        coprime monic polynomials in $\mathbb F_{16}[x]$. Then any right $\mathcal{R}$-         module in $\frac{\mathcal{R}[x]}{\langle x^n-1\rangle}$ is the sum of the               $\mathcal{R}$-submodules
        \allowdisplaybreaks \begin{gather*} \langle u^{i-1}\hat{f_j} + \langle x^n-1\rangle \rangle, \langle u^{i-1}v\hat{f_j} + \langle x^n-1\rangle \rangle, \langle u^{i-1}v^2\hat{f_j} + \langle x^n-1\rangle \rangle , \\ \langle u^{i-1}v^3\hat{f_j} + \langle x^n-1\rangle \rangle,  
        \langle (u^i+ v\beta_{1i} + v^2\beta_{2i} + v^3\beta_{3i})\hat{f_j} + \langle x^n-1\rangle \rangle ,\\ (\langle u^i \hat{f}_j + \langle x^n-1\rangle\rangle + \langle v\hat{f}_j + \langle x^n-1\rangle\rangle + \langle v^2\hat{f}_j  + \langle x^n-1\rangle\rangle + \langle v^3\hat{f}_j + \langle x^n-1\rangle\rangle)\\  ~ \text{for}~  i= \{1,\dots,k\}. \\
        \langle ((u + u^2\alpha_{13} + u^3\alpha_{14} + \cdots + u^{k-1}\alpha_{1k}) \\ + \sum_{i=1}^{3} v^i( \alpha_{(i+1)1} + u\alpha_{(i+1)2} + u^2\alpha_{(i+1)3} + \cdots + u^{k-1}\alpha_{(i+1)k} ))\hat{f_j} + \langle x^n-1\rangle \rangle ,\\
        \langle (v(u + u^2\alpha_{13} + u^3\alpha_{14} + \cdots + u^{k-1}\alpha_{1k}) \\ + \sum_{i=2}^{3} v^i( \alpha_{i1} + u\alpha_{i2} + u^2\alpha_{i3}  + \cdots + u^{k-1}\alpha_{ik} ))\hat{f_j} + \langle x^n-1\rangle \rangle , \\ \langle (v^2(u + u^2\alpha_{13} + u^3\alpha_{14} + \cdots + u^{k-1}\alpha_{1k}) \\ + v^3( \alpha_{21} + u\alpha_{22} + u^2\alpha_{23} + u^3\alpha_{24} + \cdots + u^{k-2}\alpha_{2(k-1)}))\hat{f_j} + \langle x^n-1\rangle \rangle, \\ \langle (v^3(u + u^2\alpha_{13} + u^3\alpha_{14} + \cdots + u^{k-2}\alpha_{1(k-1)}))\hat{f_j} + \langle x^n-1\rangle \rangle.
        \end{gather*}
		
        \end{corollary}
        
        \begin{proof}
        
	Follows from Theorem \ref{th1} and Theorem \ref{th2}.
    
	 \end{proof}
    
	Throughout the paper, we use $c_0+c_1x+\dots+c_{n-1}x^{n-1}$ for the corresponding      coset $c_0+c_1x+\dots+c_{n-1}x^{n-1}+\langle x^n-1\rangle$ in the ring                  $\frac{\mathcal{R}[x]}{\langle x^n-1\rangle}$. Now, we provide the main result of       this section.

	\begin{theorem}\label{th3}
		Let $\mathcal{C}$ be a cyclic code of odd length $n$ over $\mathcal{R}$. Then there exist units	$\alpha_{ij} $ and $\beta_{mn}$ in $\mathbb F_{16}[x]$ and a family of irreducible pairwise coprime monic polynomials $P_i$, for $0\leq i\leq 5k+3 ,$ in $\mathbb F_{16}[x]$ such that $\prod_{i=0}^{5k+3} P_i=x^n-1$, $\alpha_{11} = 0 $, $\alpha_{12} = 1 $ and 
  \allowdisplaybreaks\begin{gather*}
		\mathcal{C}= \bigoplus_{i=1}^{k}\langle u^{i-1}\hat{P_{i}}\rangle \bigoplus_{i=1}^{k}\langle u^{i-1}v\hat{P}_{k+i}\rangle \bigoplus_{i=1}^{k}\langle u^{i-1}v^2\hat{P}_{2k+i}\rangle  \bigoplus_{i=1}^{k}\langle u^{i-1}v^3\hat{P}_{3k+i}\rangle \\
        \bigoplus_{i=1}^{k-1} \langle (u^i+ v\beta_{1i} + v^2\beta_{2i} + v^3\beta_{3i})\hat{P}_{4k+i}\rangle   \bigoplus \langle ((u + u^2\alpha_{13} + u^3\alpha_{14} + \cdots + u^{k-1}\alpha_{1k}) \\
      + \sum_{i=1}^{3} v^i( \alpha_{(i+1)1} + u\alpha_{(i+1)2} + u^2\alpha_{(i+1)3} + u^3\alpha_{(i+1)4} + \cdots + u^{k-1}\alpha_{(i+1)k} ))\hat{P}_{5k}\rangle
     \\  \bigoplus  \langle (v(u + u^2\alpha_{13} + u^3\alpha_{14} + \cdots + u^{k-1}\alpha_{1k}) \\ + \sum_{i=2}^{3} v^i( \alpha_{i1} + u\alpha_{i2} + u^2\alpha_{i3} + \cdots + u^{k-1}\alpha_{ik} ))\hat{P}_{5k+1} \rangle
     \\  \bigoplus  \langle (v^2(u + u^2\alpha_{13} + u^3\alpha_{14} + \cdots + u^{k-1}\alpha_{1k})
     \\+ v^3( \alpha_{21} + u\alpha_{22} + u^2\alpha_{23} + \cdots + u^{k-2}\alpha_{2(k-1)}))\hat{P}_{5k+2}\rangle
     \\  \bigoplus \langle (v^3(u + u^2\alpha_{13} + u^3\alpha_{14} + \cdots + u^{k-2}\alpha_{1(k-1)}))\hat{P}_{5k+3}\rangle. \\
		\end{gather*}
    Moreover, $|\mathcal{C}|=16^{\xi}$ where
    \allowdisplaybreaks
    \begin{gather*}
    \xi = ( \sum_{i=1}^{k} (4k-4(i-1))\deg(P_i) + (3k-3(i-1))\deg(P_{k+i})  \\
     + (2k-2(i-1))\deg(P_{2k+i})  + ( \sum_{i=1}^{k} (k-(i-1))\deg(P_{3k+i})
    \\ + \sum_{i=1}^{k-1}(4k-i)\deg(P_{4k+i}) )
    + (4k-1)\deg(P_{5k}) + (3k-3)\deg(P_{5k+1}) \\ + (2k-2)\deg(P_{5k+2}) + (k-1)\deg(P_{5k+3}).\end{gather*}
  \end{theorem}
 \begin{proof}
       Let $x^n-1=f_1f_2\cdots f_m$ where $f_j~(1\leq i\leq m)$ are irreducible, pairwise co-prime monic polynomials in $\mathbb F_{16}[x]$. Then from Corollary \ref{cor1}, any right $\mathcal{R}$ module in $\frac{\mathcal{R}[x]}{\langle x^n-1\rangle}$ is the sum of the $\mathcal{R}$-submodules \allowdisplaybreaks\begin{gather*}\langle u^{i-1}\hat{f_j} + \langle x^n-1\rangle \rangle ,
         \langle u^{i-1}v\hat{f_j} + \langle x^n-1\rangle \rangle ,  \langle u^{i-1}v^2\hat{f_j} + \langle x^n-1\rangle \rangle , \\  \langle u^{i-1}v^3\hat{f_j} + \langle x^n-1\rangle \rangle ,
        \langle (u^i+ v\beta_{1i} + v^2\beta_{2i} + v^3\beta_{3i})\hat{f_j} + \langle x^n-1\rangle \rangle \\ ~\text{for}~ i = \{1,\dots,k\}.\\
        \langle ((u + u^2\alpha_{13} + u^3\alpha_{14} + \cdots + u^{k-1}\alpha_{1k}) \\+ \sum_{i=1}^{3} v^i( \alpha_{(i+1)1} + u\alpha_{(i+1)2} + u^2\alpha_{(i+1)3} + \cdots + u^{k-1}\alpha_{(i+1)k} ))\hat{f_j} + \langle x^n-1\rangle \rangle ,\\
        \langle (v(u + u^2\alpha_{13} + u^3\alpha_{14} + \cdots + u^{k-1}\alpha_{1k}) \\ +\sum_{i=2}^{3} v^i( \alpha_{i1} + u\alpha_{i2} + u^2\alpha_{i3}  + \cdots + u^{k-1}\alpha_{ik} ))\hat{f_j} + \langle x^n-1\rangle \rangle, \\  \langle (v^2(u + u^2\alpha_{13} + u^3\alpha_{14} + \cdots + u^{k-1}\alpha_{1k}) \\+ v^3( \alpha_{21} + u\alpha_{22} + u^2\alpha_{23} + u^3\alpha_{24} + \cdots + u^{k-2}\alpha_{2(k-1)}))\hat{f_j} + \langle x^n-1\rangle \rangle, \\ \langle (v^3(u + u^2\alpha_{13} + u^3\alpha_{14} + \cdots + u^{k-2}\alpha_{1(k-1)}))\hat{f_j} + \langle x^n-1\rangle \rangle.\end{gather*}
{Now, as we have assumed that all the $f_j~(1\leq i\leq m)$ are irreducible, pairwise co-prime monic polynomials in $\mathbb F_{16}[x]$, we utilize them to construct the code in the form of direct sums. To do so, without loss of generality, we assume that}
    \allowdisplaybreaks
      \begin{gather*}
		\mathcal{C}= \bigoplus_{i=1}^{k}\langle u^{i-1}\hat{f}_{m_1+ \cdots+m_i+1}\rangle \oplus\cdots\oplus \langle u^{i-1}\hat{f}_{m_1+\cdots+m_i+m_{i+1}}\rangle
  \\ \bigoplus_{i=1}^{k}\langle u^{i-1}v\hat{f}_{m_1+\cdots +m_{k+i}+1}\rangle \oplus\cdots\oplus \langle u^{i-1}v\hat{f}_{m_1+\cdots +m_{k+(i+1)}}\rangle\\
   \bigoplus_{i=1}^{k}\langle u^{i-1}v^2\hat{f}_{m_1+\cdots +m_{2k+i}+1}\rangle \oplus\cdots\oplus \langle u^{i-1}v^2\hat{f}_{m_1+\cdots +m_{2k+i+1}}\rangle\\
   \bigoplus_{i=1}^{k}\langle u^{i-1}v^3\hat{f}_{m_1+\cdots +m_{3k+i}+1}\rangle \oplus\cdots\oplus \langle u^{i-1}v^3\hat{f}_{m_1+\cdots +m_{3k+i+1}}\rangle\\
   \bigoplus_{i=1}^{k-1} \langle (u^i+ v\beta_{1i} + v^2\beta_{2i} + v^3\beta_{3i})\hat{f}_{m_1+\cdots+m_{4k+i}+1}\rangle \\
   \bigoplus \cdots \bigoplus \langle (u^i+ v\beta_{1i} + v^2\beta_{2i} + v^3\beta_{3i})\hat{f}_{m_1+\cdots+m_{4k+i+1}}\rangle \\
   \bigoplus \langle ((u + u^2\alpha_{13} + u^3\alpha_{14} + \cdots + u^{k-1}\alpha_{1k})
     \\ + \sum_{i=1}^{3} v^i( \alpha_{(i+1)1} + u\alpha_{(i+1)2} + u^2\alpha_{(i+1)3} +  \cdots + u^{k-1}\alpha_{(i+1)k} ))\hat{f}_{m_1+\cdots+m_{5k}+1}\rangle \\
      \bigoplus \cdots \bigoplus \langle ((u + u^2\alpha_{13} + u^3\alpha_{14} + \cdots + u^{k-1}\alpha_{1k})
     \\ + \sum_{i=1}^{3} v^i( \alpha_{(i+1)1} + u\alpha_{(i+1)2} + u^2\alpha_{(i+1)3} +  \cdots + u^{k-1}\alpha_{(i+1)k} ))\hat{f}_{m_1+\cdots+m_{5k+1}}\rangle \\
    \bigoplus \langle(v(u + u^2\alpha_{13} + u^3\alpha_{14} + \cdots + u^{k-1}\alpha_{1k})
     \\  + \sum_{i=2}^{3} v^i( \alpha_{i1} + u\alpha_{i2} + u^2\alpha_{i3} + u^3\alpha_{i4} + \cdots + u^{k-1}\alpha_{ik} ))\hat{f}_{m_1+\cdots+m_{5k+1}+1} \rangle \\
     \bigoplus \cdots \bigoplus  \langle(v(u + u^2\alpha_{13} + u^3\alpha_{14} + \cdots + u^{k-1}\alpha_{1k})
     \\  + \sum_{i=2}^{3} v^i( \alpha_{i1} + u\alpha_{i2} + u^2\alpha_{i3} + u^3\alpha_{i4} + \cdots + u^{k-1}\alpha_{ik} ))\hat{f}_{m_1+\cdots+m_{5k+2}} \rangle   \\
        \bigoplus \langle (v^2(u + u^2\alpha_{13} + u^3\alpha_{14} + \cdots + u^{k-1}\alpha_{1k}) \\
        + v^3( \alpha_{21} + u\alpha_{22} + u^2\alpha_{23} + u^3\alpha_{24} + \cdots + u^{k-2}\alpha_{2(k-1)} )\hat{f}_{m_1+\cdots+m_{5k+2}+1} \rangle \\
         \bigoplus \cdots \bigoplus \langle (v^2(u + u^2\alpha_{13} + u^3\alpha_{14} + \cdots + u^{k-1}\alpha_{1k}) \\
        + v^3( \alpha_{21} + u\alpha_{22} + u^2\alpha_{23} + u^3\alpha_{24} + \cdots + u^{k-2}\alpha_{2(k-1)} )\hat{f}_{m_1+\cdots+m_{5k+3}} \rangle \\
         \bigoplus \langle (v^3(u + u^2\alpha_{13} + u^3\alpha_{14} + \cdots + u^{k-2}\alpha_{1(k-1)})\hat{f}_{m_1+\cdots+m_{5k+3}+1} \rangle \\
         \bigoplus \cdots \bigoplus \langle (v^3(u + u^2\alpha_{13} + u^3\alpha_{14} + \cdots + u^{k-2}\alpha_{1(k-1)})\hat{f}_{m} \rangle,
		\end{gather*}
  where integers $m_1, m_2, m_3, \dots, m_{5k+3} \geq 0$ and $m_1+m_2+\cdots+m_{5k+3}+1 \leq m $. Let $m_0 = 0$ and $m_{5k+4}$ be the integer such that $m_1+m_2+\cdots+m_{5k+4} = m $. Now, we define \\
   $$P_0=f_{m_0+1}\cdots f_{m_0+m_1}, P_1=f_{m_0+m_1+1}\cdots f_{m_0+m_1+m_2},$$ $$P_2=f_{m_0+m_1+m_2+1}\cdots f_{m_0+m_1+m_2+m_3},P_3=f_{m_0+\cdots+m_3+1}\cdots f_{m_0+\cdots+m_4},$$  $$ \dots, P_{5k+3}=f_{m_0+\cdots+m_{5k+3}+1}\cdots f_{m}.$$
   Then, from the above construction, $P_0,P_1,\dots, P_{5k+3}$ are pairwise coprime and $\prod_{i=0}^{5k+3} P_i=x^n-1$. Hence, \\
   \begin{gather*}
		\mathcal{C}= \bigoplus_{i=1}^{k}\langle u^{i-1}\hat{P_{i}}\rangle \bigoplus_{i=1}^{k}\langle u^{i-1}v\hat{P}_{k+i}\rangle \bigoplus_{i=1}^{k}\langle u^{i-1}v^2\hat{P}_{2k+i}\rangle  \bigoplus_{i=1}^{k}\langle u^{i-1}v^3\hat{P}_{3k+i}\rangle \\
        \bigoplus_{i=1}^{k-1} \langle (u^i+ v\beta_{1i} + v^2\beta_{2i} + v^3\beta_{3i})\hat{P}_{4k+i}\rangle   \bigoplus \langle ((u + u^2\alpha_{13} + u^3\alpha_{14} + \cdots + u^{k-1}\alpha_{1k}) \\
      + \sum_{i=1}^{3} v^i( \alpha_{(i+1)1} + u\alpha_{(i+1)2} + u^2\alpha_{(i+1)3} + u^3\alpha_{(i+1)4} + \cdots + u^{k-1}\alpha_{(i+1)k} ))\hat{P}_{5k}\rangle
     \\  \bigoplus  \langle (v(u + u^2\alpha_{13} + u^3\alpha_{14} + \cdots + u^{k-1}\alpha_{1k})  \\+ \sum_{i=2}^{3} v^i( \alpha_{i1} + u\alpha_{i2} + u^2\alpha_{i3} + \cdots + u^{k-1}\alpha_{ik} ))\hat{P}_{5k+1} \rangle
     \\  \bigoplus  \langle (v^2(u + u^2\alpha_{13} + u^3\alpha_{14} + \cdots + u^{k-1}\alpha_{1k})
    \\ + v^3( \alpha_{21} + u\alpha_{22} + u^2\alpha_{23} + \cdots + u^{k-2}\alpha_{2(k-1)}))\hat{P}_{5k+2}\rangle
     \\  \bigoplus \langle (v^3(u + u^2\alpha_{13} + u^3\alpha_{14} + \cdots + u^{k-2}\alpha_{1(k-1)}))\hat{P}_{5k+3}\rangle. 
		\end{gather*}
Furthermore, the cardinalities of the individual direct sum components are as follows: We will add all the cardinalities when we write the cardinality of a whole code in the form of direct sums. 
\begin{gather*}
    |\langle u^{i-1}\hat{P_i}\rangle|={16}^{(4k-4(i-1))\deg(P_i)},
    |\langle u^{i-1}v\hat{P}_{k+i}\rangle|={16}^{(3k-3(i-1))\deg(P_{k+i})}, \\
    |\langle u^{i-1}v^2\hat{P}_{2k+i}\rangle|={16}^{(2k-2(i-1))\deg(P_{2k+i})},
    |\langle u^{i-1}v^3\hat{P}_{3k+i}\rangle|={16}^{(k-(i-1))\deg(P_{3k+i})}, \\
    |\langle (u^i+ v\beta_{1i} + v^2\beta_{2i} + v^3\beta_{3i})\hat{P}_{4k+i}\rangle| = {16}^{(4k-i)\deg(P_{4k+i})} ~ \text{For} ~ i = \{1,2,\dots,k\}. \\
    |\langle((u + u^2\alpha_{13} + u^3\alpha_{14} + \cdots + u^{k-1}\alpha_{1k})
     \\ + \sum_{i=1}^{3} v^i( \alpha_{(i+1)1} + u\alpha_{(i+1)2} + u^2\alpha_{(i+1)3} +  \cdots + u^{k-1}\alpha_{(i+1)k} ))\hat{P}_{5k}\rangle |\\ ={16}^{(4k-1)\deg(P_{5k})}, \\
     | \langle (v(u + u^2\alpha_{13} + u^3\alpha_{14} + \cdots + u^{k-1}\alpha_{1k})
     \\  + \sum_{i=2}^{3} v^i( \alpha_{i1} + u\alpha_{i2} + u^2\alpha_{i3} + u^3\alpha_{i4} + \cdots + u^{k-1}\alpha_{ik} ))\hat{P}_{5k+1} \rangle | \\ ={16}^{(3k-3)\deg(P_{5k+1})}, \\
      |\langle (v^2(u + u^2\alpha_{13} + u^3\alpha_{14} + \cdots + u^{k-1}\alpha_{1k})\\
      + v^3( \alpha_{21} + u\alpha_{22} + u^2\alpha_{23} + u^3\alpha_{24} + \cdots + u^{k-2}\alpha_{2(k-1)}))\hat{P}_{5k+2}\rangle| \\ ={16}^{(2k-2)\deg(P_{5k+2})}, \\
      | \langle (v^3(u + u^2\alpha_{13} + u^3\alpha_{14} + \cdots + u^{k-2}\alpha_{1(k-1)}))\hat{P}_{5k+3}\rangle| = {16}^{(k-1)\deg(P_{5k+3})}.
     \end{gather*}
Thus, $|\mathcal{C}|=16^{\xi}$, where
    \begin{gather*}
    \xi = ( \sum_{i=1}^{k} (4k-4(i-1))\deg(P_i) + (3k-3(i-1))\deg(P_{k+i})  \\
     + (2k-2(i-1))\deg(P_{2k+i})  + ( \sum_{i=1}^{k} (k-(i-1))\deg(P_{3k+i})
    \\ + \sum_{i=1}^{k-1}(4k-i)\deg(P_{4k+i}) )
    + (4k-1)\deg(P_{5k}) + (3k-3)\deg(P_{5k+1}) \\ + (2k-2)\deg(P_{5k+2}) + (k-1)\deg(P_{5k+3}).\end{gather*}
      \end{proof}
Now, we consider $\mathcal{R}'=\frac{\mathbb F_{16}[x]}{\langle x^n-1 \rangle_{\mathbb F_{16}[x]}}$ where ${\langle x^n-1 \rangle}_{\mathbb F_{16}[x]}$ is an ideal of $\mathbb F_{16}[x]$ generated by  $x^n-1$.
\begin{theorem}\label{th4} Let $\mathcal{C}$ be a cyclic code of odd length $n$ over $\mathcal{R}$. Then there exist polynomials
		$A_1, A_2, A_3,\\ \dots , A_{5k+3}$ in $\mathbb F_{16}[x]$ which are factors of $x^n-1$ such that
\allowdisplaybreaks\begin{gather*}
		\mathcal{C}= \bigoplus_{i=1}^{k} u^{i-1}\langle A_{i}\rangle_{\mathcal{R}'}
  \bigoplus_{i=1}^{k} u^{i-1}v\langle A_{k+i}\rangle_{\mathcal{R}'} \bigoplus_{i=1}^{k} u^{i-1}v^2\langle A_{2k+i}\rangle_{\mathcal{R}'}  \bigoplus_{i=1}^{k} u^{i-1}v^3\langle A_{3k+i}\rangle_{\mathcal{R}'} \\
        \bigoplus_{i=1}^{k-1} (u^i+ v\beta_{1i} + v^2\beta_{2i} + v^3\beta_{3i}) \langle A_{4k+i}\rangle_{\mathcal{R}'}   \bigoplus ((u + u^2\alpha_{13} + u^3\alpha_{14} + \cdots + u^{k-1}\alpha_{1k})
     \\ + \sum_{i=1}^{3} v^i( \alpha_{(i+1)1} + u\alpha_{(i+1)2} + u^2\alpha_{(i+1)3} + u^3\alpha_{(i+1)4} + \cdots + u^{k-1}\alpha_{(i+1)k} ))\langle A_{5k}\rangle_{\mathcal{R}'}
     \\  \bigoplus   (v(u + u^2\alpha_{13} + u^3\alpha_{14} + \cdots + u^{k-1}\alpha_{1k})
       \\+ \sum_{i=2}^{3} v^i( \alpha_{i1} + u\alpha_{i2} + u^2\alpha_{i3} + \cdots + u^{k-1}\alpha_{ik} ))\langle A_{5k+1}\rangle_{\mathcal{R}'}
     \\  \bigoplus  (v^2(u + u^2\alpha_{13} + u^3\alpha_{14} + \cdots + u^{k-1}\alpha_{1k})
       \\+ v^3( \alpha_{21} + u\alpha_{22} + u^2\alpha_{23} +  \cdots + u^{k-2}\alpha_{2(k-1)}))\langle A_{5k+2}\rangle_{\mathcal{R}'}
     \\  \bigoplus  (v^3(u + u^2\alpha_{13} + u^3\alpha_{14} + \cdots + u^{k-2}\alpha_{1(k-1)}))\langle A_{5k+3}\rangle_{\mathcal{R}'},
		\end{gather*}
		where $\alpha_{ij} $ are units in $\mathbb F_{16}[x]$ and $\langle - \rangle_{\mathcal{R}'} $ is an ideal of $\mathcal{R}'$ generated by $-$. Also,
		$$|\mathcal{C}|=16^{(5k+3)n-(\deg A_1+\deg A_2+ \cdots \deg A_{5k+3} }).$$
	\end{theorem}
 \begin{proof}
		From Theorem \ref{th3}, we have
  \allowdisplaybreaks
  \begin{gather*}
		\mathcal{C}= \bigoplus_{i=1}^{k}\langle u^{i-1}\hat{P_{i}}\rangle \bigoplus_{i=1}^{k}\langle u^{i-1}v\hat{P}_{k+i}\rangle \bigoplus_{i=1}^{k}\langle u^{i-1}v^2\hat{P}_{2k+i}\rangle  \bigoplus_{i=1}^{k}\langle u^{i-1}v^3\hat{P}_{3k+i}\rangle \\
        \bigoplus_{i=1}^{k-1} \langle (u^i+ v\beta_{1i} + v^2\beta_{2i} + v^3\beta_{3i})\hat{P}_{4k+i}\rangle   \bigoplus \langle ((u + u^2\alpha_{13} + u^3\alpha_{14} + \cdots + u^{k-1}\alpha_{1k}) \\
      + \sum_{i=1}^{3} v^i( \alpha_{(i+1)1} + u\alpha_{(i+1)2} + u^2\alpha_{(i+1)3} + u^3\alpha_{(i+1)4} + \cdots + u^{k-1}\alpha_{(i+1)k} ))\hat{P}_{5k}\rangle
     \\  \bigoplus  \langle (v(u + u^2\alpha_{13} + u^3\alpha_{14} + \cdots + u^{k-1}\alpha_{1k})  \\+ \sum_{i=2}^{3} v^i( \alpha_{i1} + u\alpha_{i2} + u^2\alpha_{i3} + \cdots + u^{k-1}\alpha_{ik} ))\hat{P}_{5k+1} \rangle
     \\  \bigoplus  \langle (v^2(u + u^2\alpha_{13} + u^3\alpha_{14} + \cdots + u^{k-1}\alpha_{1k})
     \\+ v^3( \alpha_{21} + u\alpha_{22} + u^2\alpha_{23} + \cdots + u^{k-2}\alpha_{2(k-1)}))\hat{P}_{5k+2}\rangle
     \\  \bigoplus \langle (v^3(u + u^2\alpha_{13} + u^3\alpha_{14} + \cdots + u^{k-2}\alpha_{1(k-1)}))\hat{P}_{5k+3}\rangle. 
		\end{gather*}
        {Now, to write the code in turms of $\mathcal{R}'$, We can easily prove that}
 \allowdisplaybreaks
  \begin{gather} \label{eq3}
		\langle \hat{P_i} \rangle =  \bigoplus_{j=1}^{k}  u^{j-1}\langle \hat{P_i} \rangle_{\mathcal{R}'}
      \bigoplus_{j=1}^{k}  u^{j-1}v\langle \hat{P_i} \rangle_{\mathcal{R}'}
      \bigoplus_{j=1}^{k}  u^{j-1}v^2\langle \hat{P_i} \rangle_{\mathcal{R}'}
      \bigoplus_{j=1}^{k}  u^{j-1}v^3\langle \hat{P_i} \rangle_{\mathcal{R}'} \nonumber\\
      \text{for} ~0 \leq i \leq 5k+3.
		\end{gather}
        {Now, we substitute the value of each $\hat{P_i}$ as given in the Equation (\ref{eq3}) in the $\mathcal{C}$. } 
   \allowdisplaybreaks
  \begin{gather*}
		\mathcal{C}= \bigoplus_{i=1}^{k} u^{i-1} ( \bigoplus_{j=1}^{(k-1)-(i-2)}  u^{j-1}\langle \hat{P_i} \rangle_{\mathcal{R}'}
      \oplus  u^{j-1}v\langle \hat{P_i} \rangle_{\mathcal{R}'}
     \\ \oplus u^{j-1}v^2\langle \hat{P_i} \rangle_{\mathcal{R}'}
      \oplus  u^{j-1}v^3\langle \hat{P_i} \rangle_{\mathcal{R}'})
      \\ \bigoplus_{i=1}^{k} u^{i-1}v\left( \bigoplus_{j=1}^{(k-1)-(i-2)}  u^{j-1}\langle \hat{P}_{k+i} \rangle_{\mathcal{R}'}
      ~\oplus~  u^{j-1}v\langle \hat{P}_{k+i} \rangle_{\mathcal{R}'}
      ~\oplus~  u^{j-1}v^2\langle \hat{P}_{k+i} \rangle_{\mathcal{R}'}\right)
      \\ \bigoplus_{i=1}^{k} u^{i-1}v^2\left( \bigoplus_{j=1}^{(k-1)-(i-2)}  u^{j-1}\langle \hat{P}_{2k+i} \rangle_{\mathcal{R}'}
      ~\oplus~  u^{j-1}v\langle \hat{P}_{2k+i} \rangle_{\mathcal{R}'}\right)
      \\ \bigoplus_{i=1}^{k} u^{i-1}v^3\left( \bigoplus_{j=1}^{(k-1)-(i-2)}  u^{j-1}\langle \hat{P}_{3k+i} \rangle_{\mathcal{R}'}\right)
       \bigoplus_{i=1}^{k-1}  (u^i+ v\beta_{1i} + v^2\beta_{2i} + v^3\beta_{3i}) \\
      ( \bigoplus_{j=1}^{(k-1)-(i-2)}  u^{j-1}\langle \hat{P}_{4k+i} \rangle_{\mathcal{R}'}
      ~\oplus~  u^{j-1}v\langle \hat{P}_{4k+i} \rangle_{\mathcal{R}'}
      \\~\oplus~  u^{j-1}v^2\langle \hat{P}_{4k+i} \rangle_{\mathcal{R}'}
      ~\oplus~  u^{j-1}v^3\langle \hat{P}_{4k+i} \rangle_{\mathcal{R}'})
      \\  \bigoplus  ((u + u^2\alpha_{13} + u^3\alpha_{14} + \cdots + u^{k-1}\alpha_{1k})
     \\ + \sum_{i=1}^{3} v^i( \alpha_{(i+1)1} + u\alpha_{(i+1)2} + u^2\alpha_{(i+1)3} + u^3\alpha_{(i+1)4} + \cdots + u^{k-1}\alpha_{(i+1)k} )) \\ \left( \bigoplus_{j=1}^{k}  u^{j-1}\langle \hat{P}_{5k} \rangle_{\mathcal{R}'}
      \oplus  u^{j-1}v\langle \hat{P}_{5k} \rangle_{\mathcal{R}'}
      \oplus u^{j-1}v^2\langle \hat{P}_{5k} \rangle_{\mathcal{R}'}
      \oplus  u^{j-1}v^3\langle \hat{P}_{5k} \rangle_{\mathcal{R}'}\right)
     \\  \bigoplus   (v(u + u^2\alpha_{13} + u^3\alpha_{14} + \cdots + u^{k-1}\alpha_{1k})
     \\  + \sum_{i=2}^{3} v^i( \alpha_{i1} + u\alpha_{i2} + u^2\alpha_{i3} + u^3\alpha_{i4} + \cdots + u^{k-1}\alpha_{ik} )) \\
     \left( \bigoplus_{j=1}^{k}  u^{j-1}\langle \hat{P}_{5k+1} \rangle_{\mathcal{R}'}
      \oplus  u^{j-1}v\langle \hat{P}_{5k+1} \rangle_{\mathcal{R}'}
      \oplus  u^{j-1}v^2\langle \hat{P}_{5k+1} \rangle_{\mathcal{R}'}\right)
     \\  \bigoplus (v^2(u + u^2\alpha_{13} + u^3\alpha_{14} + \cdots + u^{k-1}\alpha_{1k})
       \\+ v^3( \alpha_{21} + u\alpha_{22} + u^2\alpha_{23}  + \cdots + u^{k-2}\alpha_{2(k-1)})) \\
     \left( \bigoplus_{j=1}^{k}  u^{j-1}\langle \hat{P}_{5k+2} \rangle_{\mathcal{R}'}
      \oplus u^{j-1}v\langle \hat{P}_{5k+2} \rangle_{\mathcal{R}'}\right)
     \\  \bigoplus  (v^3(u + u^2\alpha_{13} + u^3\alpha_{14} + \cdots + u^{k-2}\alpha_{1(k-1)}))
     \left( \bigoplus_{j=1}^{k}  u^{j-1}\langle \hat{P}_{5k+3} \rangle_{\mathcal{R}'}\right)
		\end{gather*}
{Now, we multiply and rearrange the terms in the form of componants of $\mathcal{C}$ multiplied by associated direct sum of $\hat{P}_{i}'s$ to write the code $\mathcal{C}$ in the form of $\mathcal{R}'.$ After rearranging, we get}
 \allowdisplaybreaks
    \begin{gather*}
        \mathcal{C} = \langle \hat{P_{1}} \rangle_{\mathcal{R}'} \bigoplus_{j=1}^{k-1} u^{j} \left( \bigoplus_{i=1}^{j+1} \langle \hat{P_i} \rangle_{\mathcal{R}'} \bigoplus_{i=1}^{j} \langle \hat{P}_{4k+i} \rangle_{\mathcal{R}'} \bigoplus \langle \hat{P}_{5k} \rangle_{\mathcal{R}'}  \right) \\
         \bigoplus_{j=1}^{3} v^{j} \left( \bigoplus_{i=1}^{j} \langle \hat{P}_{ik+1} \rangle_{\mathcal{R}'} \bigoplus_{i=1}^{k-1} \langle \hat{P}_{4k+i} \rangle_{\mathcal{R}'} \bigoplus_{i=0}^{j-1} \langle \hat{P}_{5k+i} \rangle_{\mathcal{R}'} \bigoplus \langle \hat{P_{1}} \rangle_{\mathcal{R}'}  \right) \\
          \bigoplus_{j=1}^{k-1} u^{j}v \left( \bigoplus_{i=1}^{j+1} \langle \hat{P_{i}} \rangle_{\mathcal{R}'} \bigoplus_{i=1}^{j+1} \langle \hat{P}_{k+i} \rangle_{\mathcal{R}'} \bigoplus_{i=1}^{k-1} \langle \hat{P}_{4k+i} \rangle_{\mathcal{R}'} \bigoplus \langle \hat{P}_{5k} \rangle_{\mathcal{R}'} \bigoplus \langle \hat{P}_{5k+1} \rangle_{\mathcal{R}'}  \right) \\
           \bigoplus_{j=1}^{k-1} u^{j}v^2 \left( \bigoplus_{i=1}^{j+1} \langle \hat{P_{i}} \rangle_{\mathcal{R}'} \bigoplus_{i=1}^{j+1} \langle \hat{P}_{k+i} \rangle_{\mathcal{R}'}  \bigoplus_{i=1}^{j+1} \langle \hat{P}_{2k+i} \rangle_{\mathcal{R}'} \bigoplus_{i=1}^{k-1} \langle \hat{P}_{4k+i} \rangle_{\mathcal{R}'} \bigoplus_{t=0}^{2} \langle \hat{P}_{5k+t} \rangle_{\mathcal{R}'}   \right) \\
            \bigoplus_{j=1}^{k-1} u^{j}v^3 ( \bigoplus_{i=1}^{j+1} \langle \hat{P_{i}} \rangle_{\mathcal{R}'} \bigoplus_{i=1}^{j+1} \langle \hat{P}_{k+i} \rangle_{\mathcal{R}'}  \bigoplus_{i=1}^{j+1} \langle \hat{P}_{2k+i} \rangle_{\mathcal{R}'} \\ \bigoplus_{i=1}^{j+1} \langle \hat{P}_{3k+i} \rangle_{\mathcal{R}'} \bigoplus_{i=1}^{k-1} \langle \hat{P}_{4k+i} \rangle_{\mathcal{R}'} \bigoplus_{t=0}^{3} \langle \hat{P}_{5k+t} \rangle_{\mathcal{R}'}   ) \\
            \bigoplus_{j=1}^{k-1}  (u^j+ v\beta_{1i} + v^2\beta_{2i} + v^3\beta_{3i}) \left( \bigoplus_{i=1}^{j} \langle \hat{P}_{4k+i} \rangle_{\mathcal{R}'} \right)
            \\ \bigoplus  ((u + u^2\alpha_{13} + u^3\alpha_{14} + \cdots + u^{k-1}\alpha_{1k})
           \\ + \sum_{i=1}^{3} v^i( \alpha_{(i+1)1} + u\alpha_{(i+1)2} + u^2\alpha_{(i+1)3} + u^3\alpha_{(i+1)4} + \cdots + u^{k-1}\alpha_{(i+1)k} \langle \hat{P}_{5k} \rangle_{\mathcal{R}'} \\
           \bigoplus   (v(u + u^2\alpha_{13} + u^3\alpha_{14} + \cdots + u^{k-1}\alpha_{1k})
           \\ + \sum_{i=2}^{3} v^i( \alpha_{i1} + u\alpha_{i2} + u^2\alpha_{i3} +  \cdots + u^{k-1}\alpha_{ik} ))\langle \hat{P}_{5k+1} \rangle_{\mathcal{R}'} \\
           \bigoplus (v^2(u + u^2\alpha_{13} + u^3\alpha_{14} + \cdots + u^{k-1}\alpha_{1k})
     \\ + v^3( \alpha_{21} + u\alpha_{22} + u^2\alpha_{23} + \cdots + u^{k-2}\alpha_{2(k-1)})) \langle \hat{P}_{5k+2} \rangle_{\mathcal{R}'} \\
           \bigoplus  (v^3(u + u^2\alpha_{13} + u^3\alpha_{14} + \cdots + u^{k-2}\alpha_{1(k-1)})) \langle \hat{P}_{5k+3} \rangle_{\mathcal{R}'}.
           \end{gather*}
    Now, we assume that
    \allowdisplaybreaks
    \begin{gather*}
        A_1 = \hat{P_{1}}, A_{j+1} = \left( \bigoplus_{i=1}^{j+1}  \hat{P_i}  \bigoplus_{i=1}^{j}  \hat{P}_{4k+i}  \bigoplus  \hat{P}_{5k}   \right) \text{for} ~ j = \{1,\dots,k-1\}. \\
         A_{k+1} = \left( \bigoplus  \hat{P}_{k+1}  \bigoplus_{i=1}^{k-1}  \hat{P}_{4k+i}  \bigoplus  \hat{P}_{5k}  \bigoplus  \hat{P_{1}}   \right), \\
         A_{k+1+j} = \left( \bigoplus_{i=1}^{j+1}  \hat{P_{i}}  \bigoplus_{i=1}^{j+1}  \hat{P}_{k+i}  \bigoplus_{i=1}^{k-1}  \hat{P}_{4k+i}  \bigoplus  \hat{P}_{5k}  \bigoplus  \hat{P}_{5k+1}   \right) \text{for} ~ j = \{1,\dots,k-1\}.  \\
         A_{2k+1} =   \left( \bigoplus_{i=1}^{j}  \hat{P}_{ik+1}  \bigoplus_{i=1}^{k-1}  \hat{P}_{4k+i}  \bigoplus_{i=0}^{j-1}  \hat{P}_{5k+i}  \bigoplus  \hat{P_{1}}   \right) \text{for} ~ j =2. \\
         A_{2k+1+j} = \left( \bigoplus_{i=1}^{j+1}  \hat{P_{i}}  \bigoplus_{i=1}^{j+1}  \hat{P}_{k+i}   \bigoplus_{i=1}^{j+1}  \hat{P}_{2k+i} \bigoplus_{i=1}^{k-1}  \hat{P}_{4k+i}  \bigoplus_{t=0}^{2}  \hat{P}_{5k+t}    \right) \text{for} ~ j = \{1,\dots,k-1\}. \\
         A_{3k+1} =  \left( \bigoplus_{i=1}^{j}  \hat{P}_{ik+1}  \bigoplus_{i=1}^{k-1}  \hat{P}_{4k+i} \bigoplus_{i=0}^{j-1}  \hat{P}_{5k+i} \bigoplus  \hat{P_{1}}   \right) \text{for} ~ j =3. \\
         A_{3k+1+j} = \left( \bigoplus_{i=1}^{j+1} \hat{P_{i}}  \bigoplus_{i=1}^{j+1}  \hat{P}_{k+i}   \bigoplus_{i=1}^{j+1}  \hat{P}_{2k+i}  \bigoplus_{i=1}^{j+1}  \hat{P}_{3k+i}  \bigoplus_{i=1}^{k-1}  \hat{P}_{4k+i}  \bigoplus_{t=0}^{3} \hat{P}_{5k+t}    \right) \\ \text{for}~ j = \{1,\dots,k-1\}. \\
         A_{4k+j} = \left( \bigoplus_{i=1}^{j} \hat{P}_{4k+i}  \right) \text{for}~ j = \{1,\dots,k-1\}.~ A_{5k}=  \hat{P}_{5k}, A_{5k+1}=  \hat{P}_{5k+1}, \\ A_{5k+2}=  \hat{P}_{5k+2}, A_{5k+3}=   \hat{P}_{5k+3}.
    \end{gather*}
    Now, it is left to prove that
    \allowdisplaybreaks
        $$\langle A_1 \rangle_{\mathcal{R}'} = \langle \hat{P_{1}} \rangle_{\mathcal{R}'},\langle  A_{j+1} \rangle_{\mathcal{R}'} = \left( \bigoplus_{i=1}^{j+1} \langle \hat{P_i} \rangle_{\mathcal{R}'} \bigoplus_{i=1}^{j} \langle \hat{P}_{4k+i} \rangle_{\mathcal{R}'} \bigoplus \langle \hat{P}_{5k} \rangle_{\mathcal{R}'}  \right) $$  $$\text{for} ~ j = \{1,\dots,k-1\}.$$ 
        $$ \langle  A_{k+1} \rangle_{\mathcal{R}'} = \left( \bigoplus \langle \hat{P}_{k+1} \rangle_{\mathcal{R}'} \bigoplus_{i=1}^{k-1} \langle \hat{P}_{4k+i} \rangle_{\mathcal{R}'} \bigoplus \langle \hat{P}_{5k} \rangle_{\mathcal{R}'} \bigoplus \langle \hat{P_{1}} \rangle_{\mathcal{R}'}  \right), $$
       $$ \langle  A_{k+1+j} \rangle_{\mathcal{R}'} = \left( \bigoplus_{i=1}^{j+1} \langle \hat{P_{i}} \rangle_{\mathcal{R}'} \bigoplus_{i=1}^{j+1} \langle \hat{P}_{k+i} \rangle_{\mathcal{R}'} \bigoplus_{i=1}^{k-1} \langle \hat{P}_{4k+i} \rangle_{\mathcal{R}'} \bigoplus \langle \hat{P}_{5k} \rangle_{\mathcal{R}'} \bigoplus \langle \hat{P}_{5k+1} \rangle_{\mathcal{R}'}  \right) $$ $$ \text{for} ~ j = \{1,\dots,k-1\}. $$ 
       $$ \langle  A_{2k+1} \rangle_{\mathcal{R}'} =   \left( \bigoplus_{i=1}^{j} \langle \hat{P}_{ik+1} \rangle_{\mathcal{R}'} \bigoplus_{i=1}^{k-1} \langle \hat{P}_{4k+i} \rangle_{\mathcal{R}'} \bigoplus_{i=0}^{j-1} \langle \hat{P}_{5k+i} \rangle_{\mathcal{R}'} \bigoplus \langle \hat{P_{1}} \rangle_{\mathcal{R}'}  \right) \text{for} ~ j =2.$$ 
       $$ \langle  A_{2k+1+j} \rangle_{\mathcal{R}'} = \left( \bigoplus_{i=1}^{j+1} \langle \hat{P_{i}} \rangle_{\mathcal{R}'} \bigoplus_{i=1}^{j+1} \langle \hat{P}_{k+i} \rangle_{\mathcal{R}'}  \bigoplus_{i=1}^{j+1} \langle \hat{P}_{2k+i} \rangle_{\mathcal{R}'} \bigoplus_{i=1}^{k-1} \langle \hat{P}_{4k+i} \rangle_{\mathcal{R}'} \bigoplus_{t=0}^{2} \langle \hat{P}_{5k+t} \rangle_{\mathcal{R}'}   \right) $$  $$\text{for} ~ j = \{1,\dots,k-1\}. $$
       $$ \langle  A_{3k+1} \rangle_{\mathcal{R}'} =  \left( \bigoplus_{i=1}^{j} \langle \hat{P}_{ik+1} \rangle_{\mathcal{R}'} \bigoplus_{i=1}^{k-1} \langle \hat{P}_{4k+i} \rangle_{\mathcal{R}'} \bigoplus_{i=0}^{j-1} \langle \hat{P}_{5k+i} \rangle_{\mathcal{R}'} \bigoplus \langle \hat{P_{1}} \rangle_{\mathcal{R}'}  \right) \text{for} ~ j =3. $$ 
       $$ \langle  A_{3k+1+j} \rangle_{\mathcal{R}'} = ( \bigoplus_{i=1}^{j+1} \langle \hat{P_{i}} \rangle_{\mathcal{R}'} \bigoplus_{i=1}^{j+1} \langle \hat{P}_{k+i} \rangle_{\mathcal{R}'}  \bigoplus_{i=1}^{j+1} \langle \hat{P}_{2k+i} \rangle_{\mathcal{R}'} $$ $$\bigoplus_{i=1}^{j+1} \langle \hat{P}_{3k+i} \rangle_{\mathcal{R}'} \bigoplus_{i=1}^{k-1} \langle \hat{P}_{4k+i} \rangle_{\mathcal{R}'} \bigoplus_{t=0}^{3} \langle \hat{P}_{5k+t} \rangle_{\mathcal{R}'} )$$ $$ \text{for} ~ j = \{1,\dots,k-1\}. $$ 
       $$ \langle  A_{4k+j} \rangle_{\mathcal{R}'} = \left( \bigoplus_{i=1}^{j} \langle \hat{P}_{4k+i} \rangle_{\mathcal{R}'} \right) \text{for} ~ j = \{1,\dots,k-1\}.\langle  A_{5k} \rangle_{\mathcal{R}'} \\ = \langle \hat{P}_{5k} \rangle_{\mathcal{R}'}, $$  $$ \langle  A_{5k+1} \rangle_{\mathcal{R}'}= \langle \hat{P}_{5k+1} \rangle_{\mathcal{R}'} ,\langle  A_{5k+2} \rangle_{\mathcal{R}'} = \langle \hat{P}_{5k+2} \rangle_{\mathcal{R}'} ,\langle  A_{5k+3} \rangle_{\mathcal{R}'} =  \langle \hat{P}_{5k+3} \rangle_{\mathcal{R}'}.$$
    
    For any distinct $i,j\in \{1,\dots,5k+3\}$, $x^n-1| \hat{P_i}\hat{P_j}$ and $P_iP_j=0$ in $\frac{\mathcal{R}[x]}{\langle x^n-1 \rangle}$. As $\{P_j,\hat{P_j}\}$ are coprime pairs for $j=1,\dots,5k+3$, there exist $h_{0j},h_{1j}\in \mathbb F_{16}[x]$ such that $h_{0j}P_j+h_{1j}\hat{P_j}=1$.
		Further, if we consider $h_{0j}P_j+h_{1j}\hat{P_j}=1$ for $j=2,\dots,5k+3$, then there exist polynomials $w_1,w_2,\dots,w_{5k+3}$ in $\mathbb F_{16}[x]$ such that
		$$w_1P_2\cdots P_{5k+3}+w_2\hat{P_2}P_3\cdots P_{5k+3}+\cdots+ w_{5k+3}P_2P_3\cdots P_{5k+2} \hat{P}_{5k+3}=1.$$
	Multiplying the above equation by $\hat{P_1}$, we get
		$\hat{P_1}=w_1 \hat{P_1}P_2\cdots P_{5k+3}.$ Also, by the above construction,
        $w_1A_1P_2\cdots P_{5k+3} = \hat{P_1}$ which implies that $\langle A_1 \rangle_{\mathcal{R}'} = \langle \hat{P_{1}} \rangle $.
        Now, $A_2 = \hat{P_1} + \hat{P_2} + \hat{P}_{4k+1} + \hat{P}_{5k}$. Hence, $w_1A_2P_2\cdots P_{5k+3} = \hat{P_1}$ and implies that $\hat{P_1} \in \langle A_2 \rangle_{\mathcal{R}'}$. If we continue this process, we get $\hat{P_2} \in \langle A_2 \rangle_{\mathcal{R}'}$, $\hat{P}_{4k+1} \in \langle A_2 \rangle_{\mathcal{R}'}$, $\hat{P}_{5k} \in \langle A_2 \rangle_{\mathcal{R}'}$. Thus,
        $ \langle A_2 \rangle_{\mathcal{R}'} = \langle \hat{P_1}\rangle_{\mathcal{R}'} + \langle\hat{P_2}\rangle_{\mathcal{R}'} + \langle \hat{P}_{4k+1} \rangle_{\mathcal{R}'} + \langle \hat{P}_{5k}\rangle_{\mathcal{R}'}$.
        Following the same procedure, we have
        \allowdisplaybreaks
         \begin{gather*}
        \langle A_1 \rangle_{\mathcal{R}'} = \langle \hat{P_{1}} \rangle_{\mathcal{R}'},\langle  A_{j+1} \rangle_{\mathcal{R}'} = \left( \bigoplus_{i=1}^{j+1} \langle \hat{P_i} \rangle_{\mathcal{R}'} \bigoplus_{i=1}^{j} \langle \hat{P}_{4k+i} \rangle_{\mathcal{R}'} \bigoplus \langle \hat{P}_{5k} \rangle_{\mathcal{R}'}  \right) \\ \text{for} ~ j = \{2,\dots,k-1\}. \\
         \langle  A_{k+1} \rangle_{\mathcal{R}'} = \left( \bigoplus \langle \hat{P}_{k+1} \rangle_{\mathcal{R}'} \bigoplus_{i=1}^{k-1} \langle \hat{P}_{4k+i} \rangle_{\mathcal{R}'} \bigoplus \langle \hat{P}_{5k} \rangle_{\mathcal{R}'} \bigoplus \langle \hat{P_{1}} \rangle_{\mathcal{R}'}  \right), \\
        \langle  A_{k+1+j} \rangle_{\mathcal{R}'} = \left( \bigoplus_{i=1}^{j+1} \langle \hat{P_{i}} \rangle_{\mathcal{R}'} \bigoplus_{i=1}^{j+1} \langle \hat{P}_{k+i} \rangle_{\mathcal{R}'} \bigoplus_{i=1}^{k-1} \langle \hat{P}_{4k+i} \rangle_{\mathcal{R}'} \bigoplus \langle \hat{P}_{5k} \rangle_{\mathcal{R}'} \bigoplus \langle \hat{P}_{5k+1} \rangle_{\mathcal{R}'}  \right) \\ \text{for} ~ j = \{1,\dots,k-1\}.  \\
        \langle  A_{2k+1} \rangle_{\mathcal{R}'} =   \left( \bigoplus_{i=1}^{j} \langle \hat{P}_{ik+1} \rangle_{\mathcal{R}'} \bigoplus_{i=1}^{k-1} \langle \hat{P}_{4k+i} \rangle_{\mathcal{R}'} \bigoplus_{i=0}^{j-1} \langle \hat{P}_{5k+i} \rangle_{\mathcal{R}'} \bigoplus \langle \hat{P_{1}} \rangle_{\mathcal{R}'}  \right) \\ \text{for}~ j =2. \\
        \langle  A_{2k+1+j} \rangle_{\mathcal{R}'} = \left( \bigoplus_{i=1}^{j+1} \langle \hat{P_{i}} \rangle_{\mathcal{R}'} \bigoplus_{i=1}^{j+1} \langle \hat{P}_{k+i} \rangle_{\mathcal{R}'}  \bigoplus_{i=1}^{j+1} \langle \hat{P}_{2k+i} \rangle_{\mathcal{R}'} \bigoplus_{i=1}^{k-1} \langle \hat{P}_{4k+i} \rangle_{\mathcal{R}'} \bigoplus_{t=0}^{2} \langle \hat{P}_{5k+t} \rangle_{\mathcal{R}'}   \right) \\ \text{for} ~ j = \{1,\dots,k-1\}. \\
        \langle  A_{3k+1} \rangle_{\mathcal{R}'} =  \left( \bigoplus_{i=1}^{j} \langle \hat{P}_{ik+1} \rangle_{\mathcal{R}'} \bigoplus_{i=1}^{k-1} \langle \hat{P}_{4k+i} \rangle_{\mathcal{R}'} \bigoplus_{i=0}^{j-1} \langle \hat{P}_{5k+i} \rangle_{\mathcal{R}'} \bigoplus \langle \hat{P_{1}} \rangle_{\mathcal{R}'}  \right) \text{for}~ j =3. \\
        \langle  A_{3k+1+j} \rangle_{\mathcal{R}'} = ( \bigoplus_{i=1}^{j+1} \langle \hat{P_{i}} \rangle_{\mathcal{R}'} \bigoplus_{i=1}^{j+1} \langle \hat{P}_{k+i} \rangle_{\mathcal{R}'}  \bigoplus_{i=1}^{j+1} \langle \hat{P}_{2k+i} \rangle_{\mathcal{R}'} \\ \bigoplus_{i=1}^{j+1} \langle \hat{P}_{3k+i} \rangle_{\mathcal{R}'} \bigoplus_{i=1}^{k-1} \langle \hat{P}_{4k+i} \rangle_{\mathcal{R}'} \bigoplus_{t=0}^{3} \langle \hat{P}_{5k+t} \rangle_{\mathcal{R}'}   ) \\ \text{for} ~ j = \{1,\dots,k-1\}. \\
        \langle  A_{4k+j} \rangle_{\mathcal{R}'} = \left( \bigoplus_{i=1}^{j} \langle \hat{P}_{4k+i} \rangle_{\mathcal{R}'} \right) \text{for} ~ j = \{1,\dots,k-1\}.\langle  A_{5k} \rangle_{\mathcal{R}'} = \langle \hat{P}_{5k} \rangle_{\mathcal{R}'}, \\  \langle  A_{5k+1} \rangle_{\mathcal{R}'}= \langle \hat{P}_{5k+1} \rangle_{\mathcal{R}'} ,\langle  A_{5k+2} \rangle_{\mathcal{R}'} = \langle \hat{P}_{5k+2} \rangle_{\mathcal{R}'} ,\langle  A_{5k+3} \rangle_{\mathcal{R}'} =  \langle \hat{P}_{5k+3} \rangle_{\mathcal{R}'}.
    \end{gather*}

Since $|\mathbb F_{16}|=16$, $|\langle A_j \rangle_{\mathcal{R}'}|={16}^{n-\deg A_j}$ and hence
$|\mathcal{C}|=16^{(5k+3)n-( \sum_{j=1}^{5k+3}\deg A_j)}.$

\end{proof}

    \section{Dual of  Cyclic Codes over $\mathcal{R}$}

   In the study of linear codes over finite rings or fields, we define important structural properties such as self-orthogonality, dual-containing, and self-duality for a chosen inner product. A code $\mathcal{C}$ is said to be self-orthogonal if $\mathcal{C} \subseteq \mathcal{C}^\perp$, dual-containing if $\mathcal{C}^\perp \subseteq \mathcal{C}$, and self-dual if $\mathcal{C} = \mathcal{C}^\perp$, where $\mathcal{C}^\perp$ denotes the dual of $\mathcal{C}$. These notions depend on the type of dual being considered. Studying these duals is particularly important because codes that are self-orthogonal or dual-containing with respect to these inner products play a fundamental role in constructing quantum error-correcting codes. Specifically, in the Calderbank-Shor-Steane (CSS) and Hermitian constructions of quantum codes, one typically starts with a classical code satisfying these inclusion properties, which ensures that the resulting quantum code has the necessary orthogonality conditions for error correction in quantum channels. In our case, we are studying the structure of such duals. However, the precise characterization of self-orthogonal, dual-containing, and self-dual codes remains an open problem over the considered algebraic structure.
\subsection{Euclidean Dual of  Cyclic Codes over $\mathcal{R}$  }
 Euclidean inner product of $a = (a_0,a_1,\dots,a_{n-1})$ and $ b=(b_0,b_1,\dots,b_{n-1}) \in \mathcal{R}^{n}$ is defined by $ a\cdot b = a_0b_0 + a_1b_1 + \cdots + a_{n-1}b_{n-1}  $. If $ a\cdot b = 0$, then $a$ and $b$ are considered orthogonal. Recall that the dual of a linear code $\mathcal{C}$ is defined as $ \mathcal{C}^\perp = \{ b \in \mathcal{R}^n ~|~ a\cdot b =0, ~ \forall~ a \in \mathcal{C} \} $. For a polynomial $ f(x) = f_0 + f_1x + f_2x^2 + \cdots + f_{n-1}x^{n-1} $ with $ f_{n-1} \neq 0 $, its reciprocal is defined by $ f^*(x) = x^{n-1}f(x) = f_{n-1} + f_{n-2}x + \cdots + f_0x^{n-1} $. Note that $ \deg{f^*(x)} \leq \deg{f(x)} $. If $f_0 \neq 0,  ~\deg{f^*(x) = \deg{f(x)}} $.
     
    \begin{theorem}\label{th5}
    Let $\mathcal{C}$ be the cyclic code of odd length $n$ over $\mathcal{R}$ with
    \allowdisplaybreaks\begin{gather*}
		\mathcal{C}= \bigoplus_{i=1}^{k}\langle u^{i-1}\hat{P_{i}}\rangle \bigoplus_{i=1}^{k}\langle u^{i-1}v\hat{P}_{k+i}\rangle \bigoplus_{i=1}^{k}\langle u^{i-1}v^2\hat{P}_{2k+i}\rangle  \bigoplus_{i=1}^{k}\langle u^{i-1}v^3\hat{P}_{3k+i}\rangle \\
        \bigoplus_{i=1}^{k-1} \langle (u^i+ v\beta_{1i} + v^2\beta_{2i} + v^3\beta_{3i})\hat{P}_{4k+i}\rangle   \bigoplus \langle ((u + u^2\alpha_{13} + u^3\alpha_{14} + \cdots + u^{k-1}\alpha_{1k}) \\
      + \sum_{i=1}^{3} v^i( \alpha_{(i+1)1} + u\alpha_{(i+1)2} + u^2\alpha_{(i+1)3} + u^3\alpha_{(i+1)4} + \cdots + u^{k-1}\alpha_{(i+1)k} ))\hat{P}_{5k}\rangle
     \\  \bigoplus  \langle (v(u + u^2\alpha_{13} + u^3\alpha_{14} + \cdots + u^{k-1}\alpha_{1k}) \\ + \sum_{i=2}^{3} v^i( \alpha_{i1} + u\alpha_{i2} + u^2\alpha_{i3} + \cdots + u^{k-1}\alpha_{ik} ))\hat{P}_{5k+1} \rangle
     \\  \bigoplus  \langle (v^2(u + u^2\alpha_{13} + u^3\alpha_{14} + \cdots + u^{k-1}\alpha_{1k})
     \\+ v^3( \alpha_{21} + u\alpha_{22} + u^2\alpha_{23} + \cdots + u^{k-2}\alpha_{2(k-1)}))\hat{P}_{5k+2}\rangle
     \\  \bigoplus \langle (v^3(u + u^2\alpha_{13} + u^3\alpha_{14} + \cdots + u^{k-2}\alpha_{1(k-1)}))\hat{P}_{5k+3}\rangle. 
		\end{gather*}
  Then
  \begin{gather*}
    \mathcal{C}^\perp = \langle\hat{P}_{0}^*\rangle \bigoplus_{i=2}^{k}\langle u^{k-(i-1)}\hat{P}_{i}^*\rangle \bigoplus_{i=1}^{k} \langle v^3 \hat{P}_{k+i}^* \rangle + \langle u^{k-(i-1)} \hat{P}_{k+i}^* \rangle \\\bigoplus_{i=1}^{k} \langle v^2 \hat{P}_{2k+i}^* \rangle + \langle u^{k-(i-1)} \hat{P}_{2k+i}^* \rangle \\ \bigoplus_{i=1}^{k} \langle v \hat{P}_{3k+i}^* \rangle + \langle u^{k-(i-1)} \hat{P}_{3k+i}^* \rangle \bigoplus_{i=1}^{k-1} \langle u^{k-i}v^3 \hat{P}_{4k+i}^* \rangle  \bigoplus \langle u^{k-1}v^3 \hat{P}_{5k}^* \rangle\\ \bigoplus \langle v^3 \hat{P}_{5k+1}^* \rangle + \langle u^{k-1} \hat{P}_{5k+1}^* \rangle \\   \bigoplus \langle v^2 \hat{P}_{5k+2}^* \rangle + \langle u^{k-1} \hat{P}_{5k+2}^* \rangle \bigoplus \langle v \hat{P}_{5k+3}^* \rangle + \langle u^{k-1} \hat{P}_{5k+3}^* \rangle.
\end{gather*}
Moreover, $|\mathcal{C}^\perp|=16^{\eta}$ where
\begin{gather*}
  \eta = (4k)\deg {P}_0 +\sum_{i=2}^{k} 4(i-1)\deg (P_i) +\sum_{i=1}^{k} (k+3(i-1)) \deg (P_{k+i}) \\+\sum_{i=1}^{k} (2k+2(i-1)) \deg (P_{2k+i})  +\sum_{i=1}^{k} (3k+(i-1)) \deg (P_{3k+i})+\sum_{i=1}^{k-1} i \deg (P_{4k+i}) \\+ \deg (P_{5k}) + (k+3)\deg (P_{5k+1})  + (2k+2) \deg (P_{5k+2}) + (3k+1)\deg (P_{5k+3}).
\end{gather*}
\end{theorem}
\begin{proof}
    Let $\mathcal{C}$ be the cyclic code of odd length $n$ over $\mathcal{R}$. Then from Theorem \ref{th3},
    $|\mathcal{C}|=16^{\xi}$ where
\begin{gather*}
  \xi = ( \sum_{i=1}^{k} (4k-4(i-1))\deg(P_i) + \\ (3k-3(i-1))\deg(P_{k+i}) + (2k-2(i-1))\deg(P_{2k+i}) ) \\
    + \left( \sum_{i=1}^{k} (k-(i-1))\deg(P_{3k+i})
    + \sum_{i=1}^{k-1}(4k-i)\deg(P_{4k+i}) \right)\\
    + ((4k-1)\deg(P_{5k}) + (3k-3)\deg(P_{5k+1}) \\+ (2k-2)\deg(P_{5k+2}) + (k-1)\deg(P_{5k+3}) ).
\end{gather*}
Since $|\mathcal{C}||\mathcal{C}^\perp| = 16^{4kn}$ where $n= \sum_{i=1}^{5k+3} \deg(P_i), $ we have $ |\mathcal{C}^\perp| = 16^\eta $ where \begin{gather*}
    \eta = (4k)\deg {P}_0 +\sum_{i=2}^{k} 4(i-1)\deg (P_i) +\sum_{i=1}^{k} (k+3(i-1)) \deg (P_{k+i}) \\+\sum_{i=1}^{k} (2k+2(i-1)) \deg (P_{2k+i})  +\sum_{i=1}^{k} (3k+(i-1)) \deg (P_{3k+i})+\sum_{i=1}^{k-1} i \deg (P_{4k+i}) \\+ \deg (P_{5k}) + (k+3)\deg (P_{5k+1})  + (2k+2) \deg (P_{5k+2}) + (3k+1)\deg (P_{5k+3}).
\end{gather*}
Now, we consider \begin{gather*}
    \mathcal{C}^* = \langle\hat{P}_{0}^*\rangle \bigoplus_{i=2}^{k}\langle u^{k-(i-1)}\hat{P}_{i}^*\rangle \bigoplus_{i=1}^{k} \langle v^3 \hat{P}_{k+i}^* \rangle + \langle u^{k-(i-1)} \hat{P}_{k+i}^* \rangle \\\bigoplus_{i=1}^{k} \langle v^2 \hat{P}_{2k+i}^* \rangle + \langle u^{k-(i-1)} \hat{P}_{2k+i}^* \rangle \bigoplus_{i=1}^{k} \langle v \hat{P}_{3k+i}^* \rangle + \langle u^{k-(i-1)} \hat{P}_{3k+i}^* \rangle \\\bigoplus_{i=1}^{k-1} \langle u^{k-i}v^3 \hat{P}_{4k+i}^* \rangle  \bigoplus \langle u^{k-1}v^3 \hat{P}_{5k}^* \rangle \bigoplus \langle v^3 \hat{P}_{5k+1}^* \rangle + \langle u^{k-1} \hat{P}_{5k+1}^* \rangle \\   \bigoplus \langle v^2 \hat{P}_{5k+2}^* \rangle + \langle u^{k-1} \hat{P}_{5k+2}^* \rangle \bigoplus \langle v \hat{P}_{5k+3}^* \rangle + \langle u^{k-1} \hat{P}_{5k+3}^* \rangle. 
    \end{gather*}
It can be easily seen that
\allowdisplaybreaks\begin{gather*} \text{For} ~i = \{2,\dots,k\}, ~ u^{i-1}u^{k-(i-1)}\hat{P}_i\hat{P}_i^* = 0. \\ \text{For} ~ i = \{1,\dots,k\},~ u^{i-1}v (\langle v^3 \rangle + \langle u^{k-(i-1)} \rangle ) \hat{P}_{k+i}\hat{P}_{k+i}^* = 0, \\ u^{i-1}v^2 (\langle v^2 \rangle + \langle u^{k-(i-1)} \rangle ) \hat{P}_{2k+i}\hat{P}_{2k+i}^* = 0 ,~  u^{i-1}v^3 (\langle v \rangle + \langle u^{k-(i-1)} \rangle ) \\ \hat{P}_{3k+i}\hat{P}_{3k+i}^* = 0. \\
    \text{For} ~ i = \{1,\dots,k-1\}, ~  \langle (u^i+ v\alpha_{1i} + v^2\alpha_{2i} + v^3\alpha_{3i})\rangle (u^{k-i}v^3)   \hat{P}_{4k+i}\hat{P}_{4k+i}^* = 0. \\
    \langle ((u + u^2\alpha_{13} + u^3\alpha_{14} + \cdots + u^{k-1}\alpha_{1k})
     \\ + \sum_{i=1}^{3} v^i( \alpha_{(i+1)1} + u\alpha_{(i+1)2}  + \cdots + u^{k-1}\alpha_{(i+1)k} ))\rangle u^{k-1}v^3 \hat{P}_{5k}\hat{P}_{5k}^* = 0. \\
      \langle (v(u + u^2\alpha_{13} + u^3\alpha_{14} + \cdots + u^{k-1}\alpha_{1k})
     \\ + \sum_{i=2}^{3} v^i( \alpha_{i1} + u\alpha_{i2} + u^2\alpha_{i3} + u^3\alpha_{i4} + \cdots + u^{k-1}\alpha_{ik} )) \rangle (\langle v^3 \rangle + \langle u^{k-1} \rangle)\\ \hat{P}_{5k+1}\hat{P}_{5k+1}^* = 0. \\
       \langle (v^2(u + u^2\alpha_{13} + u^3\alpha_{14} + \cdots + u^{k-1}\alpha_{1k})
     \\ + v^3( \alpha_{21} + u\alpha_{22} + u^2\alpha_{23} + u^3\alpha_{24} + \cdots + u^{k-2}\alpha_{2(k-1)}))\rangle(\langle v^2 \rangle + \langle u^{k-1} \rangle)\\ \hat{P}_{5k+2}\hat{P}_{5k+2}^* = 0. \\
     \langle (v^3(u + u^2\alpha_{13} + u^3\alpha_{14} + \cdots + u^{k-2}\alpha_{1(k-1)}))\rangle (\langle v \rangle + \langle u^{k-1} \rangle)  \\\hat{P}_{5k+3}\hat{P}_{5k+3}^* = 0. \end{gather*}
     Hence, we can say that $\mathcal{C}^* \subseteq \mathcal{C}^\perp $. 
     Now,
     \allowdisplaybreaks
      \begin{gather*}
|\langle\hat{P}_{0}^*\rangle| = 16^{4k\deg(P_0)},  ~|\langle u^{k-(i-1)}\hat{P}_{i}^*\rangle| = 16^{4(i-1)\deg(P_i)} ~\text{for} ~i = \{2,\dots,k\}. \\  \text{For} ~i = \{1,\dots,k\}, ~\text{we have}~~ |\langle v^3 \hat{P}_{k+i}^* \rangle + \langle u^{k-(i-1)} \hat{P}_{k+i}^* \rangle| = 16^{k+3(i-1)\deg(P_{k+i})}, \\
     |\langle v^2 \hat{P}_{2k+i}^* \rangle + \langle u^{k-(i-1)} \hat{P}_{2k+i}^* \rangle| = 16^{2k+2(i-1)\deg(P_{2k+i})}, \\  |\langle v \hat{P}_{3k+i}^* \rangle + \langle u^{k-(i-1)} \hat{P}_{3k+i}^* \rangle| = 16^{3k+(i-1)\deg(P_{3k+i})}.
     \end{gather*}
     \begin{gather*}
     ~  \text{Next, for} ~i = \{1,\dots,k-1\}, ~\text{we have}~~ |\langle u^{k-i}v^3 \hat{P}_{4k+i}^* \rangle| = 16^{i\deg(P_{4k+i})}.\\    |\langle u^{k-1}v^3 \hat{P}_{5k}^* \rangle | = 16^{\deg(P_{5k})},   |\langle v^3 \hat{P}_{5k+1}^* \rangle + \langle u^{k-1} \hat{P}_{5k+1}^* \rangle| = 16^{k+3\deg(P_{5k+1})},   \\  |\langle v^2 \hat{P}_{5k+2}^* \rangle + \langle u^{k-1} \hat{P}_{5k+2}^* \rangle| =16^{2k+2\deg(P_{5k+2})},\\|\langle v \hat{P}_{5k+3}^* \rangle + \langle u^{k-1} \hat{P}_{5k+3}^* \rangle| = 16^{3k+1\deg(P_{5k+3})}.
    \end{gather*}
    Thus, $|\mathcal{C}^*|=16^{\eta}$ where
\begin{gather*}
    \eta = (4k)\deg {P}_0 +\sum_{i=2}^{k} 4(i-1)\deg (P_i) +\sum_{i=1}^{k} (k+3(i-1)) \deg (P_{k+i}) \\+\sum_{i=1}^{k} (2k+2(i-1)) \deg (P_{2k+i})  +\sum_{i=1}^{k} (3k+(i-1)) \deg (P_{3k+i})+\sum_{i=1}^{k-1} i \deg (P_{4k+i}) \\+ \deg (P_{5k}) + (k+3)\deg (P_{5k+1})  + (2k+2) \deg (P_{5k+2}) + (3k+1)\deg (P_{5k+3}).
\end{gather*}
Consequently, we get $ |\mathcal{C}^*|=|\mathcal{C}^\perp| $.
\end{proof}
\begin{theorem}\label{th6}
		Let $\mathcal{C}$ be a cyclic code of odd length $n$ over $\mathcal{R}$. Then there exist polynomials \\
		$A_1^*, A_2^*, A_3^*, \dots , A_{5k+3}^*$ in $\mathbb F_{16}[x],$ which are factors of $x^n-1$ such that
 \begin{gather*}
		\mathcal{C}^\perp = \bigoplus_{i=1}^{k} u^{i-1}\langle A_{i}^*\rangle_{\mathcal{R}'}
  \bigoplus_{i=1}^{k} u^{i-1}v\langle A_{k+i}^*\rangle_{\mathcal{R}'} \bigoplus_{i=1}^{k} u^{i-1}v^2\langle A_{2k+i}^*\rangle_{\mathcal{R}'}  \bigoplus_{i=1}^{k} u^{i-1}v^3\langle A_{3k+i}^*\rangle_{\mathcal{R}'} \\
        \bigoplus_{i=1}^{k-1} (u^i+ v\beta_{1i} + v^2\beta_{2i} + v^3\beta_{3i}) \langle A_{4k+i}^*\rangle_{\mathcal{R}'}   \bigoplus ((u + u^2\alpha_{13} + u^3\alpha_{14} + \cdots + u^{k-1}\alpha_{1k})
     \\ + \sum_{i=1}^{3} v^i( \alpha_{(i+1)1} + u\alpha_{(i+1)2} + u^2\alpha_{(i+1)3} + u^3\alpha_{(i+1)4} + \cdots + u^{k-1}\alpha_{(i+1)k} ))\langle A_{5k}^*\rangle_{\mathcal{R}'}
     \\  \bigoplus   (v(u + u^2\alpha_{13} + u^3\alpha_{14} + \cdots + u^{k-1}\alpha_{1k})
       \\+ \sum_{i=2}^{3} v^i( \alpha_{i1} + u\alpha_{i2} + u^2\alpha_{i3} + \cdots + u^{k-1}\alpha_{ik} ))\langle A_{5k+1}^*\rangle_{\mathcal{R}'}
     \\  \bigoplus  (v^2(u + u^2\alpha_{13} + u^3\alpha_{14} + \cdots + u^{k-1}\alpha_{1k})
      \\ + v^3( \alpha_{21} + u\alpha_{22} + u^2\alpha_{23}  + \cdots + u^{k-2}\alpha_{2(k-1)}))\langle A_{5k+2}^*\rangle_{\mathcal{R}'}
     \\  \bigoplus  (v^3(u + u^2\alpha_{13} + u^3\alpha_{14} + \cdots + u^{k-2}\alpha_{1(k-1)}))\langle A_{5k+3}^*\rangle_{\mathcal{R}'}
		\end{gather*}		
where $\alpha_{ij} $ are units in $\mathbb F_{16}[x]$ and $\langle - \rangle_{\mathcal{R}'} $ is an ideal of $\mathcal{R}'$ generated by $-$. Also,
		$$|\mathcal{C}|=16^{(5k+3)n-(\deg A_1^*+\deg A_2^*+ \cdots \deg A_{5k+3}^* }).$$

	\end{theorem}
 \begin{proof}
     It follows from Theorem \ref{th4} and Theorem  \ref{th5}.
 \end{proof}

\subsection{Hermitian Dual of Cyclic Codes over $\mathcal{R}$ }

  The Hermitian inner product of $a = (a_0,a_1,\dots,a_{n-1})$ and $ b=(b_0,b_1,\dots,b_{n-1}) \in \mathcal{R}^{*n}$ is defined by $ \langle a , b \rangle = a_0\overline{b}_0 + a_1\overline{b}_1 + \cdots + a_{n-1}\overline{b}_{n-1}  $ where ``$-$" stands for conjugation and $ \overline{b_i} = b_i^4 $. For a code $\mathcal{C}$, its Hermitian dual $\mathcal{C}^\perp$ is defined as, $ \mathcal{C}^{\perp H} = \{ b \in \mathcal{R}^n ~|~ \langle a, b \rangle =0 ~ \forall~ a \in \mathcal{C} \} $. For a polynomial $ f(x) = f_0 + f_1x + f_2x^2 + \cdots + f_{n-1}x^{n-1} $ with $ f_{n-1} \neq 0 $, its reciprocal is defined as $ f^*(x) = x^{n-1}f(x) = f_{n-1} + f_{n-2}x + \cdots + f_0x^{n-1} $. It can be seen that $ \deg{f^*(x)} \leq \deg{f(x)} $. If $f_0 \neq 0,~ \text{then}~ \deg{f^*(x) = \deg{f(x)}} $. We denote $ \overline{f(x)} = f_0^4 + f_1^4x + f_2^4x^2 + \cdots + f_{n-1}^4x^{n-1} $.

\begin{theorem} \label{th7}
    Let $\mathcal{C}$ be the cyclic code of odd length $n$ over $\mathcal{R}$ with
    \allowdisplaybreaks \begin{gather*}
		\mathcal{C}= \bigoplus_{i=1}^{k}\langle u^{i-1}\hat{P_{i}}\rangle \bigoplus_{i=1}^{k}\langle u^{i-1}v\hat{P}_{k+i}\rangle \bigoplus_{i=1}^{k}\langle u^{i-1}v^2\hat{P}_{2k+i}\rangle  \bigoplus_{i=1}^{k}\langle u^{i-1}v^3\hat{P}_{3k+i}\rangle \\
        \bigoplus_{i=1}^{k-1} \langle (u^i+ v\beta_{1i} + v^2\beta_{2i} + v^3\beta_{3i})\hat{P}_{4k+i}\rangle   \bigoplus \langle ((u + u^2\alpha_{13} + u^3\alpha_{14} + \cdots + u^{k-1}\alpha_{1k}) \\
      + \sum_{i=1}^{3} v^i( \alpha_{(i+1)1} + u\alpha_{(i+1)2} + u^2\alpha_{(i+1)3} + u^3\alpha_{(i+1)4} + \cdots + u^{k-1}\alpha_{(i+1)k} ))\hat{P}_{5k}\rangle
     \\  \bigoplus  \langle (v(u + u^2\alpha_{13} + u^3\alpha_{14} + \cdots + u^{k-1}\alpha_{1k})  \\+ \sum_{i=2}^{3} v^i( \alpha_{i1} + u\alpha_{i2} + u^2\alpha_{i3} + \cdots + u^{k-1}\alpha_{ik} ))\hat{P}_{5k+1} \rangle
     \\  \bigoplus  \langle (v^2(u + u^2\alpha_{13} + u^3\alpha_{14} + \cdots + u^{k-1}\alpha_{1k})
    \\ + v^3( \alpha_{21} + u\alpha_{22} + u^2\alpha_{23} + \cdots + u^{k-2}\alpha_{2(k-1)}))\hat{P}_{5k+2}\rangle
     \\  \bigoplus \langle (v^3(u + u^2\alpha_{13} + u^3\alpha_{14} + \cdots + u^{k-2}\alpha_{1(k-1)}))\hat{P}_{5k+3}\rangle.
		\end{gather*}
  Then
  \begin{gather*}
    \mathcal{C}^{\perp H} = \langle\hat{\overline{P}}_{0}^*\rangle \bigoplus_{i=2}^{k}\langle u^{k-(i-1)}\hat{\overline{P}}_{i}^*\rangle \bigoplus_{i=1}^{k} \langle v^3 \hat{\overline{{P}}}_{k+i}^* \rangle + \langle u^{k-(i-1)} \hat{\overline{{P}}}_{k+i}^* \rangle\\ \bigoplus_{i=1}^{k} \langle v^2 \hat{\overline{{P}}}_{2k+i}^* \rangle + \langle u^{k-(i-1)} \hat{\overline{{P}}}_{2k+i}^* \rangle  \bigoplus_{i=1}^{k} \langle v \hat{\overline{{P}}}_{3k+i}^* \rangle + \langle u^{k-(i-1)} \hat{\overline{{P}}}_{3k+i}^* \rangle \\\bigoplus_{i=1}^{k-1} \langle u^{k-i}v^3 \hat{\overline{{P}}}_{4k+i}^* \rangle  \bigoplus \langle u^{k-1}v^3 \hat{\overline{{P}}}_{5k}^* \rangle \bigoplus \langle v^3 \hat{\overline{{P}}}_{5k+1}^* \rangle + \langle u^{k-1} \hat{\overline{{P}}}_{5k+1}^* \rangle \\   \bigoplus \langle v^2 \hat{\overline{{P}}}_{5k+2}^* \rangle + \langle u^{k-1} \hat{\overline{{P}}}_{5k+2}^* \rangle \bigoplus \langle v \hat{\overline{{P}}}_{5k+3}^* \rangle + \langle u^{k-1} \hat{\overline{{P}}}_{5k+3}^* \rangle.
\end{gather*}
Further, $|\mathcal{C}^{\perp H}|=16^{\eta}$ where
\begin{gather*}
    \eta = (4k)\deg {P_0} +\sum_{i=2}^{k} 4(i-1)\deg (P_i) +\sum_{i=1}^{k} (k+3(i-1)) \deg (P_{k+i}) \\ + \sum_{i=1}^{k} (2k+2(i-1)) \deg (P_{2k+i}) + \sum_{i=1}^{k} (3k+(i-1)) \deg (P_{3k+i}) \\+\sum_{i=1}^{k-1} i \deg (P_{4k+i}) + \deg (P_{5k}) \\+ (k+3)\deg (P_{5k+1}) + (2k+2) \deg (P_{5k+2}) + (3k+1)\deg (P_{5k+3}).
\end{gather*}

\end{theorem}
\begin{proof}
    The proof follows the same procedure as the proof of Theorem \ref{th5}.
\end{proof}
\begin{theorem}\label{th8}
		Let $\mathcal{C}$ be a cyclic code of odd length $n$ over $\mathcal{R}$. Then there exist polynomials \\
		$\overline{A_1^*}, \overline{A_2^*}, \overline{A_3^*}, \dots , \overline{A_{5k+3}^*}$ in $\mathbb F_{16}[x]$ which are factors of $x^n-1$ such that
 \allowdisplaybreaks\begin{gather*}
		\mathcal{C}^{\perp H} = \bigoplus_{i=1}^{k} u^{i-1}\langle \overline{A}_{i}^*\rangle_{\mathcal{R}'}
  \bigoplus_{i=1}^{k} u^{i-1}v\langle \overline{A}_{k+i}^*\rangle_{\mathcal{R}'} \bigoplus_{i=1}^{k} u^{i-1}v^2\langle \overline{A}_{2k+i}^*\rangle_{\mathcal{R}'} \\ \bigoplus_{i=1}^{k} u^{i-1}v^3\langle \overline{A}_{3k+i}^*\rangle_{\mathcal{R}'}
       \\ \bigoplus_{i=1}^{k-1} (u^i+ v\beta_{1i} + v^2\beta_{2i} + v^3\beta_{3i}) \langle \overline{A}_{4k+i}^*\rangle_{\mathcal{R}'} \bigoplus ((u + u^2\alpha_{13} + u^3\alpha_{14} + \cdots + u^{k-1}\alpha_{1k})
     \\ + \sum_{i=1}^{3} v^i( \alpha_{(i+1)1} + u\alpha_{(i+1)2} + u^2\alpha_{(i+1)3} + u^3\alpha_{(i+1)4} + \cdots + u^{k-1}\alpha_{(i+1)k} ))\langle \overline{A}_{5k}^*\rangle_{\mathcal{R}'}
     \\  \bigoplus   (v(u + u^2\alpha_{13} + u^3\alpha_{14} + \cdots + u^{k-1}\alpha_{1k})
     \\ + \sum_{i=2}^{3} v^i( \alpha_{i1} + u\alpha_{i2} + u^2\alpha_{i3} + \cdots + u^{k-1}\alpha_{ik} ))\langle \overline{A}_{5k+1}^*\rangle_{\mathcal{R}'}
     \\  \bigoplus  (v^2(u + u^2\alpha_{13} + u^3\alpha_{14} + \cdots + u^{k-1}\alpha_{1k})
      \\ + v^3( \alpha_{21} + u\alpha_{22} + u^2\alpha_{23} + \cdots + u^{k-2}\alpha_{2(k-1)}))\langle \overline{A}_{5k+2}^*\rangle_{\mathcal{R}'}
     \\  \bigoplus  (v^3(u + u^2\alpha_{13} + u^3\alpha_{14} + \cdots + u^{k-2}\alpha_{1(k-1)}))\langle \overline{A}_{5k+3}^*\rangle_{\mathcal{R}'},
		\end{gather*}		
where $\alpha_{ij} $ are units in $\mathbb F_{16}[x]$ and $\langle - \rangle_{\mathcal{R}'} $ is an ideal of $\mathcal{R}'$ generated by $-$. Also,
		$$|\mathcal{C}|=16^{(5k+3)n-(\deg A_1^*+\deg A_2^*+ \cdots \deg A_{5k+3}^* }).$$

	\end{theorem}
 \begin{proof}
     Follows from Theorem \ref{th4} and Theorem \ref{th7}.
     
 \end{proof}
\section {Right $\mathbb F_2$-module Isometry}
\subsection{The Gray map} 
In the case of ring $\mathbb{F}_2 + u\mathbb{F}_2 $ \cite{A}, Lee weight is defined as below.
    $$ \omega_L(0)=0, \omega_L(1) = \omega_L(1+u) = 1, \omega_L(u) =2.  $$
    Also, the Gray map from $(\mathbb{F}_2 + u\mathbb{F}_2)^n $ to $ \mathbb{F}_2^{2n} $ sends $ (a + ub) $ to $ (b, a+b) $ for $ a,b \in \mathbb{F}_2^{n} $. Now, in our case we extend this Gray map $\theta $ from $\mathcal{R}$ to $ \mathbb{F}_{16}^{4kn} $ by
    \begin{gather*} \theta[(d_1 + vd_2 + v^2d_3 + v^3d_4) + u(d_5 + vd_6 + v^2d_7 + v^3d_8)  + u^2(d_9 \\ + vd_{10} + v^2d_{11} + v^3d_{12})  + \cdots + u^{k-1}(d_{4k-3} + vd_{4k-2} + v^2d_{4k-1} + v^3d_{4k})]= \\
    (d_{4k} , d_{4k-1}+d_{4k} , d_{4k-2}+d_{4k} , d_{4k-3}+d_{4k}, \dots ,d_{3}+d_{4k} , d_{2}+d_{4k}\\ , d_1 + d_2 + d_3 + \cdots + d_{4k}), \end{gather*}
where $ d_i \in \mathbb{F}_{16} $ for $ i = 1~ to~ 4k $ and $\theta$ is also distance preserving. Hence, for any element 
$$r= (d_1 + vd_2 + v^2d_3 + v^3d_4) + u(d_5 + vd_6 + v^2d_7 + v^3d_8)  + u^2(d_9 + vd_{10} + v^2d_{11} + v^3d_{12}) $$ $$ + \cdots + u^{k-1}(d_{4k-3} + vd_{4k-2} + v^2d_{4k-1} + v^3d_{4k}) $$
in $\mathcal{R}$, the extended Lee weight is defined as : 
\begin{gather*}
    \omega_L(r) = (\omega_H(d_{4k}) + \omega_H(d_{4k-1}+d_{4k}) + \omega_H(d_{4k-2}+d_{4k}) +\\\omega_H( d_{4k-3}+d_{4k})+ \cdots+ \omega_H(d_{3}+d_{4k}) + \omega_H(d_{2}+d_{4k}) + \\\omega_H(d_1 + d_2 + d_3 + \cdots + d_{4k})),
    \end{gather*} 
  where $\omega_H(.)$ denotes the Hamming weight in $\mathbb{F}_{16}$. For the $ c = (c_0,c_1, \dots, c_{n-1}) $ in $\mathcal{R}^n$, we define the following terms.
\begin{itemize}
    \item The Lee weight in $c$ is $ \omega_L(c) = \omega_L(c_0) + \omega_L(c_1) + \cdots + \omega_L(c_{n-1}) $.
    \item The Hamming weight of a codeword is $ \omega_H(c) $ is the number of nonzero components of the $c$.
\item The Hamming distance for any two codewords $c$ and $c'$ is defined as \\$ d_H(c,c') =  \omega_H(c-c') $.
 \item The Hamming distance for the code $\mathcal{C}$ is defined as \\$ d_L(\mathcal{C}) = min\{ d_L(c,c') ~|~ c \neq c'~ \forall ~c,c' \in \mathcal{C} \}$.
    \end{itemize} 
    \begin{theorem}
    Let $\mathcal{C}$ be a cyclic code over $\mathcal{R}$ of length $n$, size $M$, and Lee         distance $d$. Then $\theta(\mathcal{C})$ is a linear code over $\mathbb F_{16}$ of length      $4kn$, size $M$, and Hamming distance $d$.
\end{theorem}
\subsection{The Bachoc map}
 We have considered the ring $\mathcal{R}$ as a usual extension to the ring $ M_4(\mathbb{F}_2) $ and extended the Bachoc weight from $ M_4(\mathbb{F}_2) $ to $\mathcal{R}$, which was introduced in \cite{A}. In \cite{Pal}, the theory of cyclic codes and their duals was discussed. However, these results were derived only for odd lengths. By defining the Bachoc map, we can construct the codes of even length over $ M_4(\mathbb{F}_2) $. Let $W_B$ be the Bachoc weight over $\mathcal{R}$ and $w_B$ be the ordinary Bachoc weight of $M_4(\mathbb F_2)$ codes. For $Y_1,Y_2,\dots,Y_{k-1}\in M_4(\mathbb F_2)$, we define
\begin{gather*}	W_B(Y=Y_1+uY_2+\cdots+u^{k-1}Y_{k-1}) = w_B(Y_{k-1}) \\ +w_B(Y_{k-1}+Y_{k-2})+\cdots+w_B(Y_{k-1}+Y_{2})+w_B(Y_{1}+Y_{2}+\cdots+Y_{k-1}).\end{gather*} This immediately leads to a Bachoc map $\phi$ from $\mathcal{R}$ to $M_4^{k-1}(\mathbb F_2)$ which is defined by
	$$\phi(Y_1+uY_2+\cdots+u^{k-1}Y_{k-1})=(Y_{k-1},Y_{k-1}+Y_{k-2},\dots,Y_{k-1}+Y_{2},Y_{1}+Y_{2}+\cdots+Y_{k-1} ).$$ The map $\phi$ can be naturally extended to $\mathcal{R}^n$ by preserving isometry as,
	\begin{gather*}
	   \phi: (\mathcal{R}^n, \text{Bachoc weight})\mapsto (M_4^{(k-1)n}(\mathbb F_2), \text{Bachoc weight}).\end{gather*}
       Now, we define the map $\psi$ by
       \begin{gather*}
       \psi:M_4(\mathbb F_2[u]/ \langle u^k \rangle) \mapsto \mathbb F_{16}[u,v]/\langle u^k,v^4\rangle
       \end{gather*}

\[
\scalebox{0.72}{$
\left(\begin{smallmatrix}
\sum_{i=1}^{k}u^{i-1}(a_{1i}+b_{2i}+c_{3i}+d_{4i}) & \sum_{i=1}^{k}u^{i-1}(a_{2i}+b_{1i}+c_{4i}+d_{3i}+b_{4i}) & \sum_{i=1}^{k}u^{i-1}(a_{3i}+b_{4i}+c_{1i}+d_{2i}+c_{4i}) &\sum_{i=1}^{k}u^{i-1}(a_{4i}+b_{3i}+c_{2i}+d_{1i}+d_{4i}) \\
  \sum_{i=1}^{k}u^{i-1}(a_{2i}+b_{1i}+c_{4i}+d_{3i}) & \sum_{i=1}^{k}u^{i-1}(a_{2i}+b_{1i}+c_{4i}+d_{3i}+b_{3i}) & \sum_{i=1}^{k}u^{i-1}(a_{3i}+b_{4i}+c_{1i}+d_{2i}+c_{3i}) &  \sum_{i=1}^{k}u^{i-1}(a_{4i}+b_{3i}+c_{2i}+d_{1i}+d_{3i}) \\
  \sum_{i=1}^{k}u^{i-1}(a_{3i}+b_{4i}+c_{1i}+d_{2i}) & \sum_{i=1}^{k}u^{i-1}(a_{2i}+b_{1i}+c_{4i}+d_{3i}+b_{2i}) &  \sum_{i=1}^{k}u^{i-1}(a_{3i}+b_{4i}+c_{1i}+d_{2i}+c_{2i}) & \sum_{i=1}^{k}u^{i-1}(a_{4i}+b_{3i}+c_{2i}+d_{1i}+d_{2i}) \\
  \sum_{i=1}^{k}u^{i-1}(a_{4i}+b_{3i}+c_{2i}+d_{1i}) & \sum_{i=1}^{k}u^{i-1}(a_{2i}+b_{1i}+c_{4i}+d_{3i}+b_{1i}) & \sum_{i=1}^{k}u^{i-1}(a_{3i}+b_{4i}+c_{1i}+d_{2i}+c_{1i}) &  \sum_{i=1}^{k}u^{i-1}(a_{4i}+b_{3i}+c_{2i}+d_{1i}+d_{1i})
\end{smallmatrix}\right)
$}
\]

  \begin{gather*}
      \longmapsto\sum_{j=1}^{4} \sum_{i=1}^{k} v^{j-1}u^{i-1}(a_{ji}+b_{ji}\omega+c_{ji}\omega^2+d_{ji}\omega^3).
  \end{gather*}

  \begin{theorem}\label{th10}
       The map $\phi$ from $M_4(\mathbb F_2[u]/ \langle u^k \rangle)$ to $\mathbb F_{16}[u, v]/\langle u^k, v^4\rangle $ is a left
$F_2$-module isomorphism.
  \end{theorem}

  \begin{proof}
      Let $x, ~y \in M_4(\mathbb F_2[u]/ \langle u^k \rangle)$ and $r$ belongs to $\mathbb F_{2}$. Then  $$x=$$
\[
\scalebox{0.72}{$
\left(\begin{smallmatrix}
\sum_{i=1}^{k}u^{i-1}(a_{1i}+b_{2i}+c_{3i}+d_{4i}) & \sum_{i=1}^{k}u^{i-1}(a_{2i}+b_{1i}+c_{4i}+d_{3i}+b_{4i}) & \sum_{i=1}^{k}u^{i-1}(a_{3i}+b_{4i}+c_{1i}+d_{2i}+c_{4i}) &\sum_{i=1}^{k}u^{i-1}(a_{4i}+b_{3i}+c_{2i}+d_{1i}+d_{4i}) \\
  \sum_{i=1}^{k}u^{i-1}(a_{2i}+b_{1i}+c_{4i}+d_{3i}) & \sum_{i=1}^{k}u^{i-1}(a_{2i}+b_{1i}+c_{4i}+d_{3i}+b_{3i}) & \sum_{i=1}^{k}u^{i-1}(a_{3i}+b_{4i}+c_{1i}+d_{2i}+c_{3i}) &  \sum_{i=1}^{k}u^{i-1}(a_{4i}+b_{3i}+c_{2i}+d_{1i}+d_{3i}) \\
  \sum_{i=1}^{k}u^{i-1}(a_{3i}+b_{4i}+c_{1i}+d_{2i}) & \sum_{i=1}^{k}u^{i-1}(a_{2i}+b_{1i}+c_{4i}+d_{3i}+b_{2i}) &  \sum_{i=1}^{k}u^{i-1}(a_{3i}+b_{4i}+c_{1i}+d_{2i}+c_{2i}) & \sum_{i=1}^{k}u^{i-1}(a_{4i}+b_{3i}+c_{2i}+d_{1i}+d_{2i}) \\
  \sum_{i=1}^{k}u^{i-1}(a_{4i}+b_{3i}+c_{2i}+d_{1i}) & \sum_{i=1}^{k}u^{i-1}(a_{2i}+b_{1i}+c_{4i}+d_{3i}+b_{1i}) & \sum_{i=1}^{k}u^{i-1}(a_{3i}+b_{4i}+c_{1i}+d_{2i}+c_{1i}) &  \sum_{i=1}^{k}u^{i-1}(a_{4i}+b_{3i}+c_{2i}+d_{1i}+d_{1i})
\end{smallmatrix}\right)
$}
\]
      
   $$y=$$
\[
\scalebox{0.72}{$
\left(\begin{smallmatrix}
\sum_{i=1}^{k}u^{i-1}(a_{1i}'+b_{2i}'+c_{3i}'+d_{4i}') & \sum_{i=1}^{k}u^{i-1}(a_{2i}'+b_{1i}'+c_{4i}'+d_{3i}'+b_{4i}') & \sum_{i=1}^{k}u^{i-1}(a_{3i}'+b_{4i}'+c_{1i}'+d_{2i}'+c_{4i}') &\sum_{i=1}^{k}u^{i-1}(a_{4i}'+b_{3i}'+c_{2i}'+d_{1i}'+d_{4i}') \\
  \sum_{i=1}^{k}u^{i-1}(a_{2i}'+b_{1i}'+c_{4i}'+d_{3i}') & \sum_{i=1}^{k}u^{i-1}(a_{2i}'+b_{1i}'+c_{4i}'+d_{3i}'+b_{3i}') & \sum_{i=1}^{k}u^{i-1}(a_{3i}'+b_{4i}'+c_{1i}'+d_{2i}'+c_{3i}') &  \sum_{i=1}^{k}u^{i-1}(a_{4i}'+b_{3i}'+c_{2i}'+d_{1i}'+d_{3i}') \\
  \sum_{i=1}^{k}u^{i-1}(a_{3i}'+b_{4i}'+c_{1i}'+d_{2i}') & \sum_{i=1}^{k}u^{i-1}(a_{2i}'+b_{1i}'+c_{4i}'+d_{3i}'+b_{2i}') &  \sum_{i=1}^{k}u^{i-1}(a_{3i}'+b_{4i}'+c_{1i}'+d_{2i}'+c_{2i}') & \sum_{i=1}^{k}u^{i-1}(a_{4i}'+b_{3i}'+c_{2i}'+d_{1i}'+d_{2i}') \\
  \sum_{i=1}^{k}u^{i-1}(a_{4i}'+b_{3i}'+c_{2i}'+d_{1i}') & \sum_{i=1}^{k}u^{i-1}(a_{2i}'+b_{1i}'+c_{4i}'+d_{3i}'+b_{1i}') & \sum_{i=1}^{k}u^{i-1}(a_{3i}'+b_{4i}'+c_{1i}'+d_{2i}'+c_{1i}') &  \sum_{i=1}^{k}u^{i-1}(a_{4i}'+b_{3i}'+c_{2i}'+d_{1i}'+d_{1i}')
\end{smallmatrix}\right)
$}
\]
$$ \text{and}~ x+y =\begin{pmatrix}
      A & B & C & D \\
      E & F & G & H \\
      I & J & K & L \\
      M & N & O & P
  \end{pmatrix} $$ \\
where $A = \sum_{i=1}^{k}u^{i-1}((a_{1i}'+b_{2i}'+c_{3i}'+d_{4i}')+(a_{1i}+b_{2i}+c_{3i}+d_{4i}))$,\\
$B = \sum_{i=1}^{k}u^{i-1}((a_{2i}'+b_{1i}'+c_{4i}'+d_{3i}'+b_{4i}')+(a_{2i}+b_{1i}+c_{4i}+d_{3i}+b_{4i}))$,\\
$C = \sum_{i=1}^{k}u^{i-1}((a_{3i}'+b_{4i}'+c_{1i}'+d_{2i}'+c_{4i}')+(a_{3i}+b_{4i}+c_{1i}+d_{2i}+c_{4i}))$,\\
$D = \sum_{i=1}^{k}u^{i-1}((a_{4i}'+b_{3i}'+c_{2i}'+d_{1i}'+d_{4i}')+(a_{4i}+b_{3i}+c_{2i}+d_{1i}+d_{4i}))$,\\
$E = \sum_{i=1}^{k}u^{i-1}((a_{2i}'+b_{1i}'+c_{4i}'+d_{3i}')+(a_{2i}+b_{1i}+c_{4i}+d_{3i}))$,\\
$ F = \sum_{i=1}^{k}u^{i-1}((a_{2i}'+b_{1i}'+c_{4i}'+d_{3i}'+b_{3i}')+(a_{2i}+b_{1i}+c_{4i}+d_{3i}+b_{3i}))$, \\
$G = \sum_{i=1}^{k}u^{i-1}((a_{3i}'+b_{4i}'+c_{1i}'+d_{2i}'+c_{3i}')+(a_{3i}+b_{4i}+c_{1i}+d_{2i}+c_{3i}))$,\\
$ H = \sum_{i=1}^{k}u^{i-1}((a_{4i}'+b_{3i}'+c_{2i}'+d_{1i}'+d_{3i}')+(a_{4i}+b_{3i}+c_{2i}+d_{1i}+d_{3i}))$,\\
$I = \sum_{i=1}^{k}u^{i-1}((a_{3i}'+b_{4i}'+c_{1i}'+d_{2i}')+(a_{3i}+b_{4i}+c_{1i}+d_{2i})),$\\
$J = \sum_{i=1}^{k}u^{i-1}((a_{2i}'+b_{1i}'+c_{4i}'+d_{3i}'+b_{2i}')+(a_{2i}+b_{1i}+c_{4i}+d_{3i}+b_{2i}))$,\\
$K = \sum_{i=1}^{k}u^{i-1}((a_{3i}'+b_{4i}'+c_{1i}'+d_{2i}'+c_{2i}')+(a_{3i}+b_{4i}+c_{1i}+d_{2i}+c_{2i}))$,\\
$L = \sum_{i=1}^{k}u^{i-1}((a_{4i}'+b_{3i}'+c_{2i}'+d_{1i}'+d_{2i}')+(a_{4i}+b_{3i}+c_{2i}+d_{1i}+d_{2i}))$,\\
$M = \sum_{i=1}^{k}u^{i-1}((a_{4i}'+b_{3i}'+c_{2i}'+d_{1i}')+(a_{4i}+b_{3i}+c_{2i}+d_{1i})),$\\
$ N = \sum_{i=1}^{k}u^{i-1}((a_{2i}'+b_{1i}'+c_{4i}'+d_{3i}'+b_{1i}')+(a_{2i}+b_{1i}+c_{4i}+d_{3i}+b_{1i}))$,\\
$O = \sum_{i=1}^{k}u^{i-1}((a_{3i}'+b_{4i}'+c_{1i}'+d_{2i}'+c_{1i}')+(a_{3i}+b_{4i}+c_{1i}+d_{2i}+c_{1i}))$,\\
$P = \sum_{i=1}^{k}u^{i-1}((a_{4i}'+b_{3i}'+c_{2i}'+d_{1i}'+d_{1i}')+(a_{4i}+b_{3i}+c_{2i}+d_{1i}+d_{1i}))$.\\

Hence,\begin{align*}
\phi(x+y) =  \sum_{j=1}^{4} \sum_{i=1}^{k} v^{j-1}u^{i-1}(a_{ji}+a_{ji}') +(b_{ji}+b_{ji}')\omega+(c_{ji}+c_{ji}')\omega^2 \\ + (d_{ji}+d_{ji}')\omega^3) = \sum_{j=1}^{4} \sum_{i=1}^{k} v^{j-1}u^{i-1}(a_{ji}+b_{ji}\omega+c_{ji}\omega^2+d_{ji}\omega^3) \\ +\sum_{j=1}^{4} \sum_{i=1}^{k} v^{j-1}u^{i-1}(a_{ji}'+b_{ji}'\omega+c_{ji}'\omega^2+d_{ji}'\omega^3) = \phi(x)+\phi(y).
  \end{align*}
  Now, $rx= \begin{pmatrix}
            A' & B' & C' & D' \\
            E' & F' & G' & H' \\
            I' & J' & K' & L' \\
            M' & N' & O' & P'
            \end{pmatrix}$ \\\\
    \noindent where
    
$$ A'=\sum_{i=1}^{k}u^{i-1}r(a_{1i}+b_{2i}+c_{3i}+d_{4i}), 
B'= \sum_{i=1}^{k}u^{i-1}r(a_{2i}+b_{1i}+c_{4i}+d_{3i}+b_{4i}),$$
$$ C'= \sum_{i=1}^{k}u^{i-1}r(a_{3i}+b_{4i}+c_{1i}+d_{2i}+c_{4i}), 
D'= \sum_{i=1}^{k}u^{i-1}r(a_{4i}+b_{3i}+c_{2i}+d_{1i}+d_{4i}),$$
$$ E'= \sum_{i=1}^{k}u^{i-1}r(a_{2i}+b_{1i}+c_{4i}+d_{3i}), 
F'= \sum_{i=1}^{k}u^{i-1}r(a_{2i}+b_{1i}+c_{4i}+d_{3i}+b_{3i}),$$
$$ G'= \sum_{i=1}^{k}u^{i-1}r(a_{3i}+b_{4i}+c_{1i}+d_{2i}+c_{3i}), 
H'= \sum_{i=1}^{k}u^{i-1}r(a_{4i}+b_{3i}+c_{2i}+d_{1i}+d_{3i}),$$
$$ I'= \sum_{i=1}^{k}u^{i-1}r(a_{3i}+b_{4i}+c_{1i}+d_{2i}), 
J'= \sum_{i=1}^{k}u^{i-1}r(a_{2i}+b_{1i}+c_{4i}+d_{3i}+b_{2i}),$$ 
$$ K'= \sum_{i=1}^{k}u^{i-1}r(a_{3i}+b_{4i}+c_{1i}+d_{2i}+c_{2i}), 
L'= \sum_{i=1}^{k}u^{i-1}r(a_{4i}+b_{3i}+c_{2i}+d_{1i}+d_{2i}),$$
$$ M'=\sum_{i=1}^{k}u^{i-1}r(a_{4i}+b_{3i}+c_{2i}+d_{1i}), 
N'= \sum_{i=1}^{k}u^{i-1}r(a_{2i}+b_{1i}+c_{4i}+d_{3i}+b_{1i}),$$
$$ O'= \sum_{i=1}^{k}u^{i-1}r(a_{3i}+b_{4i}+c_{1i}+d_{2i}+c_{1i}),
P'= \sum_{i=1}^{k}u^{i-1}r(a_{4i}+b_{3i}+c_{2i}+d_{1i}+d_{1i}).$$
    
 Hence, \begin{align*}
      \phi(rx)&= \sum_{j=1}^{4} \sum_{i=1}^{k} v^{j-1}u^{i-1}r(a_{ji}+b_{ji}\omega+c_{ji}\omega^2+d_{ji}\omega^3) \\&
      = r\sum_{j=1}^{4} \sum_{i=1}^{k} v^{j-1}u^{i-1}(a_{ji}+b_{ji}\omega+c_{ji}\omega^2+d_{ji}\omega^3) \\&
      =r\phi(x).
  \end{align*}
  Now, assume that $\phi(x)=\phi(y)$.
  This implies that
  
$$ \sum_{j=1}^{4} \sum_{i=1}^{k} v^{j-1}u^{i-1}(a_{ji}+b_{ji}\omega+c_{ji}\omega^2+d_{ji}\omega^3) $$ 
$$= \sum_{j=1}^{4} \sum_{i=1}^{k} v^{j-1}u^{i-1}(a_{ji}'+b_{ji}'\omega+c_{ji}'\omega^2+d_{ji}'\omega^3).$$
  
  Recall that two polynomials of degree $n$ in indeterminate $x$ are said to be equal
if and only if their respective coefficients of $x^i$ for $i = 0,1,\dots,n$ are equal. Therefore, we have
$a_{ji}=a_{ji}' , b_{ji}=b_{ji}' , c_{ji}=c_{ji}' , d_{ji}=d_{ji}',$
which implies that $x=y$. Hence the result.\end{proof}
Now, we arrive at the following theorem.

  	\begin{theorem}
		Let $\mathcal{C}$ be a cyclic code over $\mathcal{R}$ of length $n$, size $M$, and Bachoc distance $d$. Then $\psi(\mathcal{C})$ is a linear code over $\mathbb F_{16}$ of length $4kn$, size $M$ and Hamming distance $d$.
	\end{theorem}
 \begin{proof}
     Let $x= (x_1,x_2,\dots,x_{n-1}), y= (y_1,y_2,\dots,y_{n-1}) \in \mathcal{C} $, and $\lambda \in \mathbb F_{16} $.\\
   \allowdisplaybreaks
    $  Let~ x_j = \begin{pmatrix}
      A & B & C & D \\
      E & F & G & H \\
      I & J & K & L \\
      M & N & O & P
  \end{pmatrix}  $
    and $y_j = \begin{pmatrix}
      A' & B' & C' & D' \\
      E' & F' & G' & H' \\
      I' & J' & K' & L' \\
      M' & N' & O' & P'
  \end{pmatrix}$ \\

   where 
$$A = \sum_{i=1}^{k}u^{i-1}(a_{j1i}+b_{j2i}+c_{j3i}+d_{j4i}),
B = \sum_{i=1}^{k}u^{i-1}(a_{j2i}+b_{j1i}+c_{j4i}+d_{j3i}+b_{j4i}),$$
$$C = \sum_{i=1}^{k}u^{i-1}(a_{j3i}+b_{j4i}+c_{j1i}+d_{j2i}+c_{j4i}),
D = \sum_{i=1}^{k}u^{i-1}(a_{j4i}+b_{j3i}+c_{j2i}+d_{j1i}+d_{j4i}),$$
$$E = \sum_{i=1}^{k}u^{i-1}(a_{j2i}+b_{j1i}+c_{j4i}+d_{j3i}),
F = \sum_{i=1}^{k}u^{i-1}(a_{j2i}+b_{j1i}+c_{j4i}+d_{j3i}+b_{j3i}),$$
$$G = \sum_{i=1}^{k}u^{i-1}(a_{j3i}+b_{j4i}+c_{j1i}+d_{j2i}+c_{j3i}),
H = \sum_{i=1}^{k}u^{i-1}(a_{j4i}+b_{j3i}+c_{j2i}+d_{j1i}+d_{j3i}),$$
$$I = \sum_{i=1}^{k}u^{i-1}(a_{j3i}+b_{j4i}+c_{j1i}+d_{j2i}),
J = \sum_{i=1}^{k}u^{i-1}(a_{j2i}+b_{j1i}+c_{j4i}+d_{j3i}+b_{j2i}),$$
$$K = \sum_{i=1}^{k}u^{i-1}(a_{j3i}+b_{j4i}+c_{j1i}+d_{j2i}+c_{j2i}),
L = \sum_{i=1}^{k}u^{i-1}(a_{j4i}+b_{j3i}+c_{j2i}+d_{j1i}+d_{j2i}),$$
$$M = \sum_{i=1}^{k}u^{i-1}(a_{j4i}+b_{j3i}+c_{j2i}+d_{j1i}),
N = \sum_{i=1}^{k}u^{i-1}(a_{j2i}+b_{j1i}+c_{j4i}+d_{j3i}+b_{j1i}),$$
$$O = \sum_{i=1}^{k}u^{i-1}(a_{j3i}+b_{j4i}+c_{j1i}+d_{j2i}+c_{j1i}),
P = \sum_{i=1}^{k}u^{i-1}(a_{j4i}+b_{j3i}+c_{j2i}+d_{j1i}+d_{j1i}),$$

and 
    
$$A' = \sum_{i=1}^{k}u^{i-1}(a_{j1i}'+b_{j2i}'+c_{j3i}'+d_{j4i}'),
B' = \sum_{i=1}^{k}u^{i-1}(a_{j2i}'+b_{j1i}'+c_{j4i}'+d_{j3i}'+b_{j4i}'),$$
$$C' = \sum_{i=1}^{k}u^{i-1}(a_{j3i}'+b_{j4i}'+c_{j1i}'+d_{j2i}'+c_{j4i}'),$$$$
D' = \sum_{i=1}^{k}u^{i-1}(a_{j4i}'+b_{j3i}'+c_{j2i}'+d_{j1i}'+d_{j4i}'),$$
$$E' = \sum_{i=1}^{k}u^{i-1}(a_{j2i}'+b_{j1i}'+c_{j4i}'+d_{j3i}'),
F' = \sum_{i=1}^{k}u^{i-1}(a_{j2i}'+b_{j1i}'+c_{j4i}'+d_{j3i}'+b_{j3i}'),$$
$$G' = \sum_{i=1}^{k}u^{i-1}(a_{j3i}'+b_{j4i}'+c_{j1i}'+d_{j2i}'+c_{j3i}'),$$$$
H' = \sum_{i=1}^{k}u^{i-1}(a_{j4i}'+b_{j3i}'+c_{j2i}'+d_{j1i}'+d_{j3i}'),$$
$$I' = \sum_{i=1}^{k}u^{i-1}(a_{j3i}'+b_{j4i}'+c_{j1i}'+d_{j2i}'), 
J' = \sum_{i=1}^{k}u^{i-1}(a_{j2i}'+b_{j1i}'+c_{j4i}'+d_{j3i}'+b_{j2i}'),$$
$$K' = \sum_{i=1}^{k}u^{i-1}(a_{j3i}'+b_{j4i}'+c_{j1i}'+d_{j2i}'+c_{j2i}'),
$$$$L' = \sum_{i=1}^{k}u^{i-1}(a_{j4i}'+b_{j3i}'+c_{j2i}'+d_{j1i}'+d_{j2i}'),$$
$$M' = \sum_{i=1}^{k}u^{i-1}(a_{j4i}'+b_{j3i}'+c_{j2i}'+d_{j1i}'),
N' = \sum_{i=1}^{k}u^{i-1}(a_{j2i}'+b_{j1i}'+c_{j4i}'+d_{j3i}'+b_{j1i}'),$$
$$O' = \sum_{i=1}^{k}u^{i-1}(a_{j3i}'+b_{j4i}'+c_{j1i}'+d_{j2i}'+c_{j1i}'),
$$$$P' = \sum_{i=1}^{k}u^{i-1}(a_{j4i}'+b_{j3i}'+c_{j2i}'+d_{j1i}'+d_{j1i}').$$

  Then by Theorem \ref{th10}, $\phi(x_j + y_j)= \phi(x_j) + \phi(y_j)$.
$$  \phi(x_j+y_j) = $$
$$ \sum_{p=1}^{4} \sum_{i=1}^{k} v^{p-1}u^{i-1}(a_{jpi}+a_{jpi}')+ $$ $$(b_{jpi}+b_{jpi}')\omega+(c_{jpi}+c_{jpi}')\omega^2+(d_{jpi}+d_{jpi}')\omega^3.$$
  Hence, \allowdisplaybreaks\begin{gather*}
      \theta(\phi(x_j+y_j)) =\\ ((a_{j4k}+a_{j4k}')+(b_{j4k}+b_{j4k}')\omega+(c_{j4k}+c_{j4k}')\omega^2+(d_{j4k}+d_{j4k}')\omega^3 ,\\ (a_{j4k}+a_{j4k}'+a_{j4(k-1)}+a_{j4(k-1)}')+(b_{j4k}+b_{j4k}'+b_{j4(k-1)}+b_{j4(k-1)}')\omega\\ +(c_{j4k}+c_{j4k}'+c_{j4(k-1)}+c_{j4(k-1)}')\omega^2+(d_{j4k}+d_{j4k}'+d_{j4(k-1)}+d_{j4(k-1)}')\omega^3  ,\\ (a_{j4k}+a_{j4k}'+a_{j4(k-2)}+a_{j4(k-2)}')+(b_{j4k}+b_{j4k}'+b_{j4(k-2)}+b_{j4(k-2)}')\omega\\ +(c_{j4k}+c_{j4k}'+c_{j4(k-2)}+c_{j4(k-2)}')\omega^2+(d_{j4k}+d_{j4k}'+d_{j4(k-2)}+d_{j4(k-2)}')\omega^3 ,\\
      \cdots ,  \sum_{p=1}^{4} \sum_{i=1}^{k} (a_{jpi}+a_{jpi}')+(b_{jpi}+b_{jpi}')\omega+(c_{jpi}+c_{jpi}')\omega^2+(d_{jpi}+d_{jpi}')\omega^3 ).
      \\  = ((a_{j4k})+(b_{j4k})\omega+(c_{j4k})\omega^2+(d_{j4k})\omega^3 ,(a_{j4k}+a_{j4(k-1)})+(b_{j4k}+b_{j4(k-1)})\omega\\ +(c_{j4k}+c_{j4(k-1)})\omega^2+(d_{j4k}+d_{j4(k-1)})\omega^3  , (a_{j4k}+a_{j4(k-2)}+(b_{j4k}+b_{j4(k-2)})\omega\\ +(c_{j4k}+c_{j4(k-2)})\omega^2+(d_{j4k}+d_{j4(k-2)})\omega^3 ,\\
      \dots ,  \sum_{p=1}^{4} \sum_{i=1}^{k} (a_{jpi})+(b_{jpi})\omega+(c_{jpi})\omega^2+(d_{jpi})\omega^3 )) \\
      + ((a_{j4k}')+(b_{j4k}')\omega+(c_{j4k}')\omega^2+(d_{j4k}')\omega^3 ,(a_{j4k}'+a_{j4(k-1)}')+(b_{j4k}'+b_{j4(k-1)}')\omega\\ +(c_{j4k}'+c_{j4(k-1)}')\omega^2+(d_{j4k}'+d_{j4(k-1)}')\omega^3  , (a_{j4k}'+a_{j4(k-2)}'+(b_{j4k}'+b_{j4(k-2)}')\omega\\ +(c_{j4k}'+c_{j4(k-2)}')\omega^2+(d_{j4k}'+d_{j4(k-2)}')\omega^3 ,\\
      \dots ,  \sum_{p=1}^{4} \sum_{i=1}^{k} (a_{jpi}')+(b_{jpi}')\omega+(c_{jpi}')\omega^2+(d_{jpi}')\omega^3 )).
      \\ =  \theta(\phi(x_j)) +  \theta(\phi(y_j)).
      \end{gather*}
      \begin{gather*}
      \text{Now,}~\theta(\phi(\lambda x_j)) = (\lambda((a_{j4k})+(b_{j4k})\omega+(c_{j4k})\omega^2+(d_{j4k})\omega^3) ,\\ \lambda((a_{j4k}+a_{j4(k-1)})+(b_{j4k}+b_{j4(k-1)})\omega \\+(c_{j4k}+c_{j4(k-1)})\omega^2+(d_{j4k}+d_{j4(k-1)})\omega^3)  ,\\ \lambda((a_{j4k}+a_{j4(k-2)}+(b_{j4k}+b_{j4(k-2)})\omega \\+(c_{j4k}+c_{j4(k-2)})\omega^2+(d_{j4k}+d_{j4(k-2)})\omega^3)) ,\\
      \cdots ,  \lambda\sum_{p=1}^{4} \sum_{i=1}^{k} (a_{jpi})+(b_{jpi})\omega+(c_{jpi})\omega^2+(d_{jpi})\omega^3 )).
      \\= \lambda (((a_{j4k})+(b_{j4k})\omega+(c_{j4k})\omega^2+(d_{j4k})\omega^3 , (a_{j4k}+a_{j4(k-1)})+(b_{j4k}+b_{j4(k-1)})\omega\\ +(c_{j4k}+c_{j4(k-1)})\omega^2+(d_{j4k}+d_{j4(k-1)})\omega^3  , (a_{j4k}+a_{j4(k-2)})+(b_{j4k}+b_{j4(k-2)})\omega\\ +(c_{j4k}+c_{j4(k-2)})\omega^2+(d_{j4k}+d_{j4(k-2)})\omega^3 ,\\
      \cdots ,  \sum_{p=1}^{4} \sum_{i=1}^{k} (a_{jpi})+(b_{jpi})\omega+(c_{jpi})\omega^2+(d_{jpi})\omega^3 )).
      \\= \lambda\theta(\phi(x_j)).
  \end{gather*}\allowdisplaybreaks
  Thus, $\theta(\phi(x))$ is linear. Note that, $\theta(\phi(x))$ is bijective. Hence, $| \mathcal{C} |= | \theta(\phi(\mathcal{C})) |$ and $d=d_H$.\\
 \end{proof}
\section{Examples }\allowdisplaybreaks

 One of the basic problems for a linear code $[n, k]$ is to optimize d and obtain $d_{MDS} = n - k + 1$. In the following examples, we have obtained linear codes that are near to MDS as $ d_{MDS} - d \leq 3$ by using the structure of $\mathcal{R}$. Here, we determine the parameters of the codes using Magma Computation Software \cite{Bosma}.

\begin{example}\allowdisplaybreaks
    For $n=5$, we have $x^5-1=(x+1)(x+w^3)(x+w^{6})(x+w^9)(x+w^{12})\in\mathbb F_{16}[x]$. Let $f_1=x+1,~f_2=x+w^3,~f_3=x+w^{6}, f_4=x+w^9,f_5=x+w^{12}$. If $k=1$, we have some cyclic codes of length $5$ over $\mathcal{R}$, and their Gray images are as follows.

{\renewcommand{\arraystretch}{1.2}{\begin{longtable}{ | l | c | c | }\hline
 Generators of cyclic codes & Images under $\theta$ & $d_{MDS}$ \\\hline	
 $\langle (w^3v^3 + w^5v^2 + w^4v + w^{13} +w )f_2 + wf_1\rangle$  &  $[20,14,5]$ & $7$ \\
	$\langle (w v^3+w^2v^2+w^5)f_2,(w^2 v^2+w v+w^5)f_1f_3f_4f_5\rangle$  &  $[20,16,3]$ & $5$ \\
	$\langle (w v^3+w^2v^2+w^5)f_2,(w^2v^3+w v +w^4 )f_1f_3f_4f_5\rangle$ &  $[20,15,4]$& $6$  \\ 
	$\langle (w v+w^2v^2+vw^5)f_2,(w^2v^2+vw^7)f_3f_4f_5\rangle$ & $[20,19,2]$ & $2$  \\\hline
\end{longtable}}}
   
\end{example}

\begin{example}
    For $n=5$ and $k=3$, we have the following cyclic codes of length $5$ over $\mathcal{R}$ and their Gray images are as follows.
        
		{\renewcommand{\arraystretch}{1.2}\begin{longtable}{| l | c | c | }

		\hline
			    	Generators of cyclic codes & Images under $\theta$
                      & $d_{MDS}$ \\
					\hline	
					$\langle (w^5u^2v^3 + vw^6)f_2,(w^5uv^3 + v^3w^6)f_3f_4f_5,$
                      &  $[60,54,4]$& $7$  \\
                        $ (w^5uv + uv^3w^6 + v^3w^9  )f_1f_2f_3f_5\rangle $    & &  \\
					$\langle (w^5u^2v^3 + vw^6)f_2,(w^7u^2v^3 + w^3v^2)f_3f_4f_5,$ &
                        $[60,56,3]$
                      & $5$ \\
                       $(wu^2v^3 + w^2u^2)f_2f_3f_1f_5 + (wu + w^5 + wv $ & &  \\
                   $+ wv^3 + wv^2)f_4f_3f_1f_5\rangle $ & &  \\
					$\langle (w^5u^2v^3 + w^5u^2v^2)f_2,(w^5u^2v^3 + w^7u^2)f_3f_4f_5,$ &  $[60,57,2]$ & $4$  \\
                       $(w^5u^2v^3 + w^5vu + w^5v^3)f_2f_3f_5f_4\rangle$ & &  \\
                      \hline
		\end{longtable}}
\end{example}
\begin{example}
		For $n=3$, the factorization of $x^3-1$ is given by $x^3-1=(x+1)(x+w^5)(x+w^{10})\in \mathbb F_{16}[x]$. Let $f_1=x+1,~f_2=x+w^5,~f_3=x+w^{10}$. Then, some good codes under the Gray images of cyclic codes of length $n = 3, k = 4$ over $\mathcal{R}$ are as follows.
   
    \allowdisplaybreaks
    {\renewcommand{\arraystretch}{1.2}\begin{longtable}{ | l | c | c | }
					\hline
			    	Generators of cyclic codes & Images under $\theta$
                      & $d_{MDS}$ \\
					\hline	
                	$\langle(wu^3v^3 + w^3u^3v + w^5u^2v^3 + w^7u^2v + w^{11}uv^3$  &  $[48,41,5]$
                      & $8$ \\
                   $+ w^{13}uv + w^3v^3 + w^7v )f_2 + (wu^3v^3 + w^5u^3v +$  & &  \\
                    $w^6u^2v^3 + w^7u^2v + w^{10}uv^3 + w^{13}uv + w^3v^3 $ & &  \\
                      $+ w^7v + w^{13})f_1, (wu^2v^2 + w^4u^3 + w^6u^2v^2 +$       & &  \\
                     $ w^8u^2 + w^{10}uv^2 + w^{14}u + w^2v^2 + w^4 )f_3\rangle$ & &  \\
					$\langle (u^3v^3+vw)f_2,(w u^3v^3+v^2+w^2)f_3,$  &  $[48,43,3]$
                      & $6$ \\
                       $(u^3v^2w^3+u^2v^2w^3+u^2w^4+w^4)f_1f_3\rangle$ & &  \\
					$\langle (u^3v^3+vw)f_2,(w u^3v^3+v^2+w^2)f_3,$
                      &  $[48,42,4]$& $7$  \\
                        $  (w^9u^2v^3 +w^9uv^3+w^8u^2+w^8u^3+w^3vu^2 $    & &  \\
                       $+w^7vu^3+w^7v^2u^2+w^3v^2u^3)f_1f_2\rangle$ & &  \\ 
					$\langle (u^3v^2w+u^3vw^3)f_2,(u^3v^2w^2+u^3vw^7)f_3,$ &
                        $[48,46,2]$
                      & $3$ \\
                       $(u^3v^2w^{13}+u^3vw^5+w^7)f_2f_1\rangle$ & &  \\
				   \hline
		\end{longtable}}
\end{example}

 \begin{example}
      For $n=7$, the factorization of $x^7-1$ is given by $x^7-1=(x+1)(x^3 + x + 1)(x^3 + x^2 + 1)\in \mathbb F_{16}[x]$. Let $f_1=x+1,~f_2=x^3 + x + 1,~f_3=x^3 + x^2 + 1$. For $n=7,k=1$, and $n=3,k=1$, {we compare the newly obtained codes over $\mathcal{R}=M_4 (\mathbb F_2 [u]/ \langle u^k \rangle)$ and existing codes for the structure $M_4(\mathbb{F}_2)$ \cite{Pal}.} For $n=3$, the factorization of $x^3-1$ is used as given in Example 3.

       \renewcommand{\arraystretch}{1.2}\begin{longtable}{  | p{6cm} | c | c |} \hline
        Generators of cyclic code $(n=7,k=1)$ & Obtained codes & Existing codes \cite{Pal} \\\hline
        $\langle (w v^3+w^2v^2+w^3v+w^7)f_1,$ & $[28,27,2]$ & $[28,25,2]$ \\
           $ (w^{11} v^3+w^{13}v^2+w^{15}v)f_3\rangle$     & & \\

        $\langle (w v^3+w^2v^2+w^3v)f_2, $ & $[28,22,4]$ & $[28,14,3]$ \\
        $(w^{3} v^3+w^{9}v)f_1f_3\rangle $ & & \\
        $\langle (w v^3+w^2v^2+w^3v)f_2, (w^{3} v^3 $ & $[28,23,3]$ & $[28,16,3]$ \\
        $+w^{9}v)f_1f_3, (w^5v^3+w v^2+w^{11}v)f_2f_3\rangle$ & & \\
          \hline
          Generators of cyclic code $(n=3,k=1)$ &  &  \\\hline
        $\langle (w v^3+w^2v^2 + w^5)f_2, (w^3v^3+w)f_1f_3\rangle$ & $[12,7,4]$ & $[12,6,3]$ \\
        $\langle (w v^3+w)f_2, (w^{3}v+w^{3}v^2)f_1f_3\rangle$ & $[12,8,3]$ & $[12,8,2]$ \\ \hline
       \end{longtable}
 \end{example}

 \section{Conclusion}
In 2022, Patel et al. \cite{s} discussed the structure of cyclic codes over $M_4(\mathbb F_2 + u\mathbb F_2 )$. \cite{s} for odd length. Initially, we have extended their work to the matrix ring $\mathcal{R}= M_4(\mathbb F_2+u\mathbb F_2+\cdots+u^{k-1}\mathbb F_2 )$. Then, we have derived the structure of cyclic codes, which includes the formation of ideals of the ring $\mathcal{R}$, cyclic codes as direct sums of submodules, and the cardinality of cyclic codes over the ring $\mathcal{R}$, respectively. Here, we have studied the Euclidean and Hermitian duals of the derived cyclic codes. Further, in the module isometry for the given ring, we have defined the Bachoc map and the Gray map, which takes our cyclic code over $\mathcal{R}$ to $\mathbb F_{16}$. In addition, we have provided some non-trivial examples of linear codes over $\mathbb F_{16}$ with good parameters that support our derived results, as well as a comparison of some existing codes for the structure $M_4(\mathbb{F}_2)$ in \cite{Pal} and obtained codes over $\mathbb F_{16}$ for the extended structure $M_4 (\mathbb F_2 [u]/ \langle u^k \rangle)$. Here, the study of self-dual, self-orthogonal, and dual-containing codes over $\mathcal{R}$ will be an interesting problem to be discussed. One can also attempt to derive a deterministic criterion which will prove that our extended family of codes is also a good choice for transmission through the MIMO channel, as given in the case of $M_4(\mathbb{F}_2)$ in \cite{F}.
    

    \section*{Acknowledgements}
The first and third authors are thankful to the Indian Institute of Technology Patna and the Department of Science and Technology (DST), Govt. of India (under SERB File Number: MTR/2022/001052, vide Diary No / Finance No SERB/F/8787/2022-2023 dated 29 December 2022), respectively, for providing financial support.

\section*{Data Availability}
The authors declare that [the/all other] data supporting the findings of this study are available within the article. Any clarification may be requested from corresponding author provided it is essential. \\
\textbf{Competing interests}: The authors declare that there is no conflict of interest regarding the publication of this manuscript.\
\

\end{document}